\documentclass[a4paper]{article}

\AtBeginDocument{%
  \providecommand\BibTeX{{%
    \normalfont B\kern-0.5em{\scshape i\kern-0.25em b}\kern-0.8em\TeX}}}

\usepackage[T1]{fontenc}
\usepackage[llparenthesis,rrparenthesis,llbracket,rrbracket]{stmaryrd}

\newcommand\infer[3]{\ensuremath{{\dfrac{#3}{#2}}{\scriptstyle #1}}}
\newcommand{\hastypei}[2]{\ensuremath{\Gamma \vdash #1 : #2}}
\newcommand{\hastypeg}[3]{\ensuremath{\Gamma, #1 \vdash #2 : #3}}
\newcommand\TVar{\mathcal{TV}ar}
\newcommand\Var{\mathcal{V}ar}
\newcommand{\pf}[1]{\textsl{PF}(#1)}
\newcommand{\conj}[1]{\ensuremath{\bigwedge(#1)}}
\newcommand{\jud}[3]{\ensuremath{#1 \vdash #2 : #3}}
\newcommand\pair[2]{\ensuremath{\langle #1,#2 \rangle}}
\newcommand{\subst}[3]{\ensuremath{[#1:=#2](#3)}}
\newcommand{\substa}[3]{\ensuremath{[#1:=#2]#3}}
\newcommand{\app}[2]{\ensuremath{#1 #2}}
\newcommand{\abs}[2]{\ensuremath{\lambda #1 . #2}}
\newcommand{\tapp}[2]{\ensuremath{#1[#2]}}
\newcommand{\pApp}[2]{\ensuremath{#1[#2]}}
\newcommand{\pAbs}[2]{\ensuremath{\Lambda #1 . #2}}
\newcommand{\proj}[2]{\ensuremath{\pi_{#1}(#2)}}
\newcommand{\proja}[2]{\ensuremath{\pi_{#1}(#2)}}
\newcommand{\ax}{\ensuremath{(\text{\textsl{ax}})}}
\newcommand{\conji}{\ensuremath{(\wedge_i)}}
\newcommand{\imple}{\ensuremath{(\Rightarrow_e)}}
\newcommand{\req}{\ensuremath{(\equiv)}}
\newcommand{\hastype}[2]{\ensuremath{#1:#2}}
\newcommand{\timpl}[2]{\ensuremath{#1 \Rightarrow #2}}
\newcommand{\tconj}[2]{\ensuremath{#1 \wedge #2}}
\newcommand{\tfor}[2]{\ensuremath{\forall #1 . #2}}
\newcommand{\sep}{\ |\ }
\newcommand{\re}{\ensuremath{\hookrightarrow}}
\newcommand{\ftva}[1]{\ensuremath{FTV}(#1)}
\newcommand{\kk}[3]{\ensuremath{K#1^{#2}_{#3}}}
\newcommand{\ka}[4]{\ensuremath{K#1^{#2}_{#3}\fhole{#4}}}
\newcommand\T{\ensuremath{{\mathcal T}}}
\newcommand\fhole[1]{\ensuremath{\llparenthesis{#1}\rrparenthesis}}
\newcommand\hole[1]{\ensuremath{\llparenthesis\,\rrparenthesis^{#1}}}
\newcommand\interp[1]{\ensuremath{\left\llbracket{#1}\right\rrbracket}}
\newcommand\SN{\ensuremath{\mathsf{SN}}}
\newcommand\eq{\ensuremath{\rightleftarrows}}
\newcommand\nneq{\ensuremath{\not\rightleftarrows}}
\newcommand\toreq{\ensuremath{\rightarrow}}
\newcommand\rulelabel[1]{\mbox{\scriptsize\sc{#1}}}

\usepackage{balance}
\usepackage{amsmath,amssymb,amsthm}
\usepackage{hyperref}%
\hypersetup{
draft=false,
colorlinks=true,
linkcolor=black,
anchorcolor=black,
citecolor=black,
urlcolor=black,
filecolor=black,
breaklinks=true}

\newtheorem{theorem}{Theorem}[section]
\newtheorem{lemma}[theorem]{Lemma}
\newtheorem{corollary}[theorem]{Corollary}
\theoremstyle{definition}
\newtheorem{definition}[theorem]{Definition}
\newtheorem{example}[theorem]{Example}

\usepackage{multicol}
\usepackage{ebproof}
\usepackage{needspace}
\usepackage{adjustbox}
\usepackage{fullpage}
\newcommand\xrecap[4]{{\noindent\sffamily\bfseries}{\bf #1 \ref{#3}}{ (#2)}{\bf.} \emph{#4}}
\newcommand\recap[3]{{\noindent\sffamily\bfseries}{\bf #1 \ref{#2}.} \emph{#3}}
\newcommand{\conje}{\ensuremath{(\wedge_e)}}
\newcommand{\impli}{\ensuremath{(\Rightarrow_i)}}
\newcommand{\fori}{\ensuremath{(\forall_i)}}
\newcommand{\fore}{\ensuremath{(\forall_e)}}
\newcommand{\p}[2]{\ensuremath{\pair{#1}{#2}}}
\newcommand{\fvt}[1]{\ensuremath{FV}(#1)}
\newcommand{\deductiona}[2]{\item \ensuremath{#1} \hfill (#2)}
\newcommand{\deductionb}[3]{\item \ensuremath{#1 \\ #2} \hfill (#3)}
\newcommand{\deductionc}[4]{\item \ensuremath{#1 \\ #2 \\ #3} \hfill (#4)}
\newcommand{\deductionsb}{\begin{enumerate}}
\newcommand{\deductionse}{\end{enumerate}}
\newcommand{\judi}[2]{\jud{\Gamma}{#1}{#2}}
\newcommand{\judg}[3]{\jud{\Gamma, #1}{#2}{#3}}
\allowdisplaybreaks

\begin{document}

\title{Polymorphic System I}

\author{Cristian F. Sottile$^1$ \and Alejandro D\'iaz-Caro$^{1,2}$ \and Pablo E. Mart\'inez L\'opez$^2$ \\
  {\small $^1$ Instituto de Investigaci\'on en Ciencias de la Computaci\'on (ICC).} \\ {\small CONICET--Universidad de Buenos Aires. Argentina.} \\
  {\small $^2$ Departamento~de Ciencia y Tecnolog\'ia.} \\ 
  {\small Universidad Nacional de Quilmes. Argentina.}}

\date{}
\maketitle

\begin{abstract}
  System I is a simply-typed lambda calculus with pairs, extended with an
  equational theory obtained from considering the type isomorphisms as equalities.
  In this work we propose an extension of System I to polymorphic types, adding
  the corresponding isomorphisms. We provide non-standard proofs of subject
  reduction and strong normalisation, extending those of System I.
\end{abstract}


\section{Introduction}
Two types $A$ and $B$ are considered isomorphic ($\equiv$) if there exist two
functions $f$ of type $A\Rightarrow B$ and $g$ of type $B\Rightarrow A$ such
that the composition $g\circ f$ is semantically equivalent to the identity in
$A$ and the composition $f\circ g$ is semantically equivalent to the identity in
$B$. Di~Cosmo et al.~\cite{DiCosmo95} characterised the isomorphic types in
different systems: simple types, simple types with pairs, polymorphism, etc.
Using this characterisation, System I has been
defined~\cite{DiazcaroDowekFSCD19}. It is a simply-typed lambda calculus with
pairs, where isomorphic types are considered equal. In this way, if $A$ and $B$
are isomorphic, every term of type $A$ can be used as a corresponding term of
type $B$. For example, the currying isomorphism $(A\wedge B)\Rightarrow C \equiv
A\Rightarrow B\Rightarrow C$ allows passing arguments one by one to a function
expecting a pair. Normally, this would imply for a function $f:(A\wedge
B)\Rightarrow C$ to be transformed through a term $r$ into $rf:A\Rightarrow
B\Rightarrow C$. System I goes further, by considering that $f$ has both types
$(A\wedge B)\Rightarrow C$ and $A\Rightarrow B\Rightarrow C$, and so, the
transformation occurs implicitly without the need for the term $r$. To make this
idea work, System I includes an equivalence between terms; for example:
$r\langle s,t \rangle\rightleftarrows rst$, since if $r$ expects a pair, it can
also take each component at a time. Also, $\beta$-reduction has to be
parametrized by the type: if the expected argument is a pair, then $r\langle s,t
\rangle$ $\beta$-reduces; otherwise, it does not $\beta$-reduce, but $rst$ does.
For example, $(\lambda x^{A\wedge B}.u)\langle r,s \rangle$ $\beta$-reduces if
$r$ has type $A$ and $s$ has type $B$. Instead, $(\lambda x^A.u)\langle r,s
\rangle$ does not reduce directly, but since it is equivalent to $(\lambda
x^A.u)rs$, which does reduce, then it also reduces, modulo this equivalence.

The idea of identifying some propositions has already been investigated, for
example, in Martin-L\"of's type theory \cite{MartinLof84}, in the Calculus of
Constructions \cite{CoquandHuetIC88}, and in Deduction modulo theory
\cite{DowekHardinKirchnerJAR03,DowekWernerJSL98}, where definitionally
equivalent propositions, for instance $A \subseteq B$, $A \in\mathcal P(B)$, and
$\forall x~(x \in A \Rightarrow x \in B)$ can be identified. But definitional
equality does not handle isomorphisms. For example, $A \wedge B$ and $B \wedge
A$ are not identified in these logics. Besides definitional equality,
identifying isomorphic types in type theory is also a goal of the univalence
axiom \cite{HoTT}. From the programming perspective, isomorphisms capture the
computational meaning correspondence between types. Taking currying again, for
example, we have a function $f$ of type $\timpl{\tconj{A}{B}}{C}$ that can be
transformed, because there exists an isomorphism, into a function $f'$ of type
$\timpl{A}{\timpl{B}{C}}$. These two functions differ in how they can be
combined with other terms, but they share a purpose: they both compute the same
value of type $C$ given two arguments of types $A$ and $B$. In this sense,
System I's proposal is to allow a programmer to focus on the meaning of
programs, combining any term with the ones that are combinable with its
isomorphic counterparts (e.g. $fx^Ay^B$ and $f'\pair{x^A}{y^B}$), ignoring the
rigid syntax of terms within the safe context provided by type isomorphisms.
From the logic perspective, isomorphisms make proofs more natural. For instance,
to prove $(A\wedge(A\Rightarrow B))\Rightarrow B$ in natural deduction we need
to introduce the conjunctive hypothesis $A\wedge (A\Rightarrow B)$, which has to
be decomposed into $A$ and $A\Rightarrow B$, while using currying allows to
transform the goal to $A\Rightarrow (A\Rightarrow B)\Rightarrow B$ and to
directly introduce the hypotheses $A$ and $A\Rightarrow B$, completely
eliminating the need for the conjunctive hypotheses.

One of the pioneers on using isomorphisms in programming languages has been
Rittri~\cite{RittriCADE90}, who used the types, equated by isomorphisms, as
search keys in program libraries.

An interpreter of a preliminary version of System I extended with a recursion
operator has been implemented in Haskell~\cite{DiazcaroMartinezlopezIFL15}. Such
a language have peculiar characteristics. For example, using the existing
isomorphism between $A\Rightarrow (B\wedge C)$ and $(A\Rightarrow
B)\wedge(A\Rightarrow C)$, we can project a function computing a pair of
elements, and obtain, through evaluation, a simpler function computing only one
of the elements of the pair, discarding the unused code that computes the output
that is not of interest to us. That paper includes some non trivial examples.

In this work we propose an extension of System I to polymorphism, considering
the corresponding isomorphisms.

\paragraph{Plan of the paper.}
The paper is organised as follows: Section~\ref{sec:SystemI} introduces the
proposed system, and Section~\ref{sec:ej1} gives examples to better clarify the
constructions. Section~\ref{sec:SR} proves the Subject Reduction property and
Section~\ref{sec:SN} the Strong Normalisation property, which are the main
theorems in the paper. Finally, Section~\ref{sec:conclusiones} discusses some
design choices, as well as possible directions for future work.

%

\section{Intuitions and Definitions}\label{sec:SystemI}
We define Polymorphic System I (PSI) as an extension of System
I~\cite{DiazcaroDowekFSCD19} to polymorphic types. The syntax of types coincides
with that of System F~\cite[Chapter~11]{Girard89} with pairs:
$$
A \quad:=\quad X \sep \timpl{A}{A} \sep \tconj{A}{A} \sep \forall X.A
$$
where $X\in\TVar$, a set of type variables.

The extension with respect to System F with pairs consists of adding a typing
rule such that if $t$ has type $A$ and $A\equiv B$, then $t$ has also type $B$,
which is valid for every pair of isomorphic types $A$ and $B$. This non-trivial
addition induces a modification of the operational semantics of the calculus.

There are eight isomorphisms characterising all the valid isomorphisms of System
F with pairs (cf.~\cite[Table~1.4]{DiCosmo95}). From those eight, we consider
the six given as a congruence in Table~\ref{tab:isos}, where $FTV(A)$ is the set
of free type variables defined as usual.

\begin{table}[t]
  \centering
  \setcounter{equation}{0}
  \begin{align}
    A \land B & \equiv B \land A\label{iso:comm} \\
    A \land (B \land C) & \equiv (A \land B) \land C\label{iso:asso} \\
    A \Rightarrow (B \land C) & \equiv (A \Rightarrow B) \land (A \Rightarrow C)\label{iso:dist} \\
    (A \land B) \Rightarrow C & \equiv A \Rightarrow B \Rightarrow C\label{iso:curry} \\
    \mbox{\scriptsize If $X \not \in FTV(A)$,}\hspace{2mm} \forall X . (A \Rightarrow B) & \equiv A \Rightarrow \forall X . B \label{iso:pcommi} \\
    \forall X .(A \land B) & \equiv \forall X . A \land \forall X . B \label{iso:pdist}
  \end{align}
  \caption{Isomorphisms considered in PSI}
  \label{tab:isos}
\end{table}
The two non-listed isomorphisms are the following:
\begin{equation}
  \label{iso:alpha}
  \forall X.A  \equiv \forall Y.[X:=Y]A
\end{equation}
\begin{equation}
  \label{iso:swap}
  \forall X.\forall Y.A \equiv \forall Y.\forall X.A
\end{equation}
The isomorphism \eqref{iso:alpha} is in fact an $\alpha$-equivalence, and we
indeed consider terms and types modulo $\alpha$-equivalence. We simply do not
make this isomorphism explicit in order to avoid confusion. The isomorphism
\eqref{iso:swap}, on the other hand, is not treated in this paper because PSI is
presented in Church style (as System I), and so, being able to swap the
arguments of a type abstraction would imply swapping the typing arguments with a
cumbersome notation and little gain. We discuss this in
Section~\ref{sec:swap}.

The added typing rule for isomorphic types induces certain equivalences between
terms. In particular, the isomorphism \eqref{iso:comm} implies that the pairs
$\pair rs$ and $\pair sr$ are indistinguishable, since both are typed as
$A\wedge B$ and also as $B\wedge A$, independently of which term has type $A$
and which one type $B$. Therefore, we consider that those two pairs are
equivalent. In the same way, as a consequence of isomorphism \eqref{iso:asso},
$\pair r{\pair st}$ is equivalent to $\pair{\pair rs}t$.

Such an equivalence between terms implies that the usual projection, which is
defined with respect to the position (i.e.~$\pi_i(\pair{r_1}{r_2})\re r_i$), is
not well-defined in this system. Indeed, $\pi_1(\pair rs)$ would reduce to $r$,
but since $\pair rs$ is equivalent to $\pair sr$, it would also reduce to $s$.
Therefore, PSI (as well as System I), defines the projection with respect to a
type: if $\Gamma\vdash r:A$, then $\proja{A}{\pair{r}{s}}\re r$.

This rule turns PSI into a non-deterministic (and therefore non-confluent)
system. Indeed, if both $r$ and $s$ have type $A$, then $\proja A{\pair rs}$
reduces non-deterministically to $r$ or to $s$. This non-determinism, however,
can be argued not to be a major problem: if we think of PSI as a proof
system, then the non-determinism, as soon as we have type preservation, implies
that the system identifies different proofs of isomorphic propositions (as a form
of proof-irrelevance). On the other hand, if PSI is thought as a programming
language, then the determinism can be recovered by the following encoding: if
$r$ and $s$ have the same type, it suffices to encode the deterministic
projection of $\pair rs$ into $r$ as
\(
  \pi_{B\Rightarrow A}(\pair{\lambda x^B.r}{\lambda x^C.s})t
\)
where $B\not\equiv C$ and $t$ has type $B$. Hence, the non-determinism of System
I (inherited in PSI) is considered a feature and not a flaw
(cf.~\cite{DiazcaroDowekFSCD19} for a longer discussion).

Thus, PSI (as well as System I) is one of the many non-deterministic calculi in
the literature,
e.g.~\cite{BoudolIC94,BucciarelliEhrhardManzonettoAPAL12,deLiguoroPipernoIC95,DezaniciancagliniDeliguoroPipernoSIAM98,PaganiRonchidellaroccaFI10}
and so our pair-construction operator can also be considered as the parallel
composition operator of a non-deterministic calculus.

In non-deterministic calculi, the non-deterministic choice is such that if $ r$
and $ s$ are two $\lambda$-terms, the term $ r\oplus s$ represents the
computation that runs either $ r$ or $ s$ non-deterministically, that is such
that $( r \oplus s) t$ reduces either to $ r t$ or $ s t$. On the other hand,
the parallel composition operator $\parallel$ is such that the term $(
r\parallel s) t$ reduces to $ r t\parallel s t$ and continue running both $ r t$
and $ s t$ in parallel. In our case, given $ r$ and $ s$ of type $A\Rightarrow
B$ and $ t$ of type $A$, the term $\pi_B(\pair rs t)$ is equivalent to
$\pi_B(\pair{r t}{st})$, which reduces to $ r t$ or $ s t$, while the term $
\pair{rt}{st}$ itself would run both computations in parallel. Hence, our
pair-constructor is equivalent to the parallel composition while the
non-deterministic choice $\oplus$ is decomposed into the pair-constructor
followed by its destructor.

In PSI and System I, the non-determinism comes from the interaction of two
operators, $\pair{}{}$ and $\pi$. This is also related to the algebraic
calculi~\cite{ArrighiDiazcaroLMCS12,ArrighiDiazcaroValironIC17,ArrighiDowekLMCS17,VauxMSCS09,DiazcaroDowekTPNC17,DiazcaroGuillermoMiquelValironLICS19},
some of which have been designed to express quantum algorithms. There is a clear
link between our pair constructor and the projection $\pi$, with the
superposition constructor $+$ and the measurement $\pi$ on these algebraic
calculi. In these cases, the pair $s + t$ is not interpreted as a
non-deterministic choice, but as a superposition of two processes running $s$
and $t$, and the operator $\pi$ is the projection related to the measurement,
which is the only non-deterministic operator. In such calculi, the
distributivity rule $(r+ s) t\rightleftarrows r t+ s t$ is seen as the
point-wise definition of the sum of two functions.
\medskip

The syntax of terms is then similar to that of System F with pairs, but with the
projections depending on types instead of position, as discussed:
\[
  r\quad:=\quad x^A\mid \lambda x^A.r\mid rr\mid \pair{r}{r}\mid \pi_A(r)\mid
  \Lambda X.r\mid \tapp{r}{A}
\]
where $x^A\in\Var$, a set of typed variables. We omit the type of variables when it is
evident from the context. For example, we write $\lambda x^A.x$ instead of
$\lambda x^A.x^A$.

The type system of PSI is standard, with only two modifications with respect to
that of System F with pairs: the projection $(\wedge_e)$, and the added rule for
isomorphisms ($\equiv$). The full system is shown in Table~\ref{fig:typing}. We write
$\Gamma\vdash r:A$ to express that $r$ has type $A$ in context $\Gamma$. Notice,
however, that since the system is given in Church-style (i.e.~variables have
their type written), the context is
redundant~\cite{GeuversKrebbersMcKinnaWiedijkLFMTP10,ParkSeoParkLeeJAR13}.
Hence, we may write ``$r$ has type $A$'' with no ambiguity. From now on, except
where indicated, we use the first upper-case letters of the Latin alphabet ($A,
B, C, \dots$) for types, the last upper-case letters of the Latin alphabet ($W,
X, Y, Z$) for type variables, lower-case Latin letters ($r, s, t, \dots$)
for terms, the last lower-case letters of the Latin alphabet ($x, y, z$) for
term variables, and upper-case Greek letters ($\Gamma, \Delta, \dots$) for
contexts.

\begin{table}[t]
  \[
    \begin{array}{c}
      \infer{{\ax}}{\Gamma,x:A\vdash x:A}{\phantom{x:A}}
      \\[3ex]
      \infer{{(\equiv)}}{\Gamma\vdash r:B}{\Gamma\vdash r:A \qquad A\equiv B}
      \\[3ex]
      \infer{{(\Rightarrow_i)}}{\Gamma\vdash\lambda x^A.r:A\Rightarrow B}{\Gamma,x:A\vdash r:B}
      \\[3ex]
      \infer{{(\Rightarrow_e)}}{\Gamma\vdash rs:B}{\Gamma\vdash r:A\Rightarrow B \qquad \Gamma\vdash s:A}
      \\[3ex]
      \infer{{(\wedge_i)}}{\Gamma\vdash \pair{r}{s}:A\wedge B}{\Gamma\vdash r:A \qquad \Gamma\vdash s:B}
      \\[3ex]
      \infer{{(\wedge_e)}}{\Gamma\vdash\pi_A(r):A}{\Gamma\vdash r:A\wedge B}
      \\[3ex]
      \infer{{(\forall_i)}}{\Gamma\vdash\Lambda X.r:\forall X.A}{\Gamma\vdash r:A \qquad X\notin FTV(\Gamma)}
      \\[3ex]
      \infer{{(\forall_e)}}{\Gamma\vdash \tapp{r}{B}:[X:=B]A}{\Gamma\vdash r:\forall X.A}
    \end{array}
  \]
  \caption{Typing rules}
  \label{fig:typing}
\end{table}

In the same way as isomorphisms~\eqref{iso:comm} and \eqref{iso:asso} induce the
commutativity and associativity of pairs, as well as a modification in the
elimination of pairs (i.e.~the projection), the isomorphism \eqref{iso:dist}
induces that an abstraction of type $A\Rightarrow (B\wedge C)$ can be considered
as a pair of abstractions of type $(A\Rightarrow B)\wedge(A\Rightarrow C)$, and
so it can be projected. Therefore, an abstraction returing a pair is identified
with a pair of abstractions, and a pair applied distributes its argument---that
is, $\lambda x^A.\pair{r}{s}\rightleftarrows \pair{\lambda x^A.r}{\lambda
x^A.s}$, and $\pair{r}{s}t\rightleftarrows \pair{\app{r}{t}}{\app{s}{t}}$, where
$\rightleftarrows$ is a symmetric symbol (and $\rightleftarrows^*$ its
transitive closure).

In addition, isomorphism \eqref{iso:curry} induces the following equivalence:
$r\pair{s}{t}\rightleftarrows rst$. However, this equivalence produces an
ambiguity with the $\beta$-reduction. For example, if $s$ has type $A$ and $t$
has type $B$, the term $(\lambda x^{\tconj{A}{B}}.r)\pair{s}{t}$ can
$\beta$-reduce to $[x~:=~\pair{s}{t}]r$, but also, since this term is equivalent
to $(\lambda x^{\tconj{A}{B}}.r)st$, which $\beta$-reduces to $([x:=s]r)t$,
reduction would not be stable by equivalence. To ensure the stability of
reduction through equivalence, the $\beta$-reduction must be performed only when the
type of the argument is the same as the type of the abstracted variable: if
$\hastypei{s}{A}$, then $(\lambda x^A.r)s\re [x:=s]r$.

The two added isomorphisms for polymorphism (\eqref{iso:pcommi} and
\eqref{iso:pdist}) also add several equivalences between terms. Two induced by
\eqref{iso:pcommi}, and four induced by \eqref{iso:pdist}.

Summarising, the operational semantics of PSI is given by the relation $\re$
modulo the symmetric relation $\rightleftarrows$. That is, we consider the
relation
\[
  \to\quad :=\quad\rightleftarrows^*\circ\re\circ\rightleftarrows^*
\]
As usual, we write $\to^*$ the reflexive and transitive closure of $\to$.
We also may write $\re^n$ to express $n$ steps in relation $\re$, and $\re_R$ to
specify that the used rule is $R$. Both relations for PSI are given in
Table~\ref{tab:SemOp}.

\begin{table}[t]
  \centering
  \begin{align*}
    \pair{r}{s} & \rightleftarrows \pair{s}{r}
                  \tag{\text{\tiny{COMM}}} \label{si:eq:comm}\\
    \pair{r}{\pair{s}{t}} & \rightleftarrows \pair{\pair{r}{s}}{t}
                            \tag{\text{\tiny{ASSO}}} \label{si:eq:asso}\\
    \abs{x^{A}}{\pair{r}{s}} & \rightleftarrows \pair{\abs{x^{A}}{r}}{\abs{x^{A}}{s}}
                               \tag{\text{\tiny{DIST$_\lambda$}}} \label{si:eq:dist-abs}\\
    \app{\pair{r}{s}}{t} & \rightleftarrows \pair{\app{r}{t}}{\app{s}{t}}
                           \tag{\text{\tiny{DIST$_{\text{\tiny{app}}}$}}} \label{si:eq:dist-app}\\
    \app{r}{\pair{s}{t}} & \rightleftarrows \app{\app{r}{s}}{t}
                           \tag{\text{\tiny{CURRY}}} \label{si:eq:curry} \\
    \mbox{\scriptsize If $ X \notin \ftva{A}, $}\hspace{2mm}
    \Lambda X . \lambda x^A . r & \rightleftarrows \lambda x^A . \Lambda X . r
                                  \tag{\text{{\tiny{P-COMM$_{\forall_i\Rightarrow_i}$}}}} \label{spcommii} \\
    \mbox{\scriptsize If $ X \notin \ftva{A}, $}\hspace{2mm}
    \pApp{(\abs{x^A}{r})}B & \rightleftarrows \abs{x^A}{\pApp{r}B}
                             \tag{\text{{\tiny{P-COMM$_{\forall_e\Rightarrow_i}$}}}} \label{spcommei} \\
    \Lambda X . \pair{ r }{ s } &  \rightleftarrows \pair{ \Lambda X . r }{ \Lambda X . s }
                                  \tag{\text{\tiny{P-DIST$_{\forall_i\land_i}$}}} \label{spdistii} \\
    \pair{ r }{ s }[A] &  \rightleftarrows \pair{ \tapp{r}{A} }{ \tapp{s}{A} }
                         \tag{\text{\tiny{P-DIST$_{\forall_e\land_i}$}}} \label{spdistei} \\
    \pi_{\forall X . A}(\Lambda X . r) &  \rightleftarrows \Lambda X . \pi_A (r)
                                         \tag{\text{\tiny{P-DIST$_{\forall_i\land_e}$}}} \label{spdistie} \\
    \mbox{\scriptsize If $ \hastype{r}{\tfor{X}{(\tconj{B}{C})}}, $}\hspace{2mm}
    \pApp{(\proja{\tfor{X}{B}}{r})}A & \rightleftarrows \proj{\substa{X}{A}{B}}{\pApp{r}A}
                                       \tag{\text{\tiny{P-DIST$_{\forall_e\land_e}$}}} \label{spdistee} \\
    \\\notag
    \mbox{If $ \hastypei{s}{A}, $}\hspace{2mm}
    (\abs{ x^A }{ r }) s & \hookrightarrow \substa{x}{s}{r} 
                         \tag{\text{\small{$\beta_{\lambda}$}}} \label{red:smallbeta} \\
    (\Lambda X.r)[A] & \hookrightarrow [X:=A]r
                       \tag{\text{\small{$\beta_\Lambda$}}} \label{red:bigbeta} \\
    \mbox{If $ \hastypei{r}{A},$}\hspace{2mm}
    \pi_A (\pair{ r }{ s }) & \hookrightarrow r 
                          \tag{\text{\small{$\pi$}}} \label{red:pi}
  \end{align*}
  \[
    \infer{}{\lambda x^A.r\rightleftarrows\lambda x^A.s}{r\rightleftarrows s}
    \quad
    \infer{}{rt\rightleftarrows st}{r\rightleftarrows s}
    \quad
    \infer{}{tr\rightleftarrows ts}{r\rightleftarrows s}
  \]
  \[
    \infer{}{\pair{t}{r}\rightleftarrows \pair{t}{s}}{r\rightleftarrows s}
    \quad
    \infer{}{\pi_A(r)\rightleftarrows\pi_A(s)}{r\rightleftarrows s}
    \quad
    \infer{}{\pair{r}{t}\rightleftarrows \pair{s}{t}}{r\rightleftarrows s}
  \]
  \[
    \infer{}{\Lambda X.r\rightleftarrows\Lambda X.s}{r\rightleftarrows s}
    \quad
    \infer{}{\tapp{r}{A}\rightleftarrows \tapp{s}{A}}{r\rightleftarrows s}
  \]\vspace{-1mm}
  \[
    \infer{}{\lambda x^A.r\re\lambda x^A.s}{r\re s}
    \quad
    \infer{}{rt\re st}{r\re s}
    \quad
    \infer{}{tr\re ts}{r\re s}
  \]\vspace{-1mm}
  \[
    {\infer{}{\pair{t}{r}\re \pair{t}{s}}{r\re s}}
    \quad
    {\infer{}{\pi_A(r)\re\pi_A(s)}{r\re s}}
    \quad
    \infer{}{\pair{r}{t}\re \pair{s}{t}}{r\re s}
  \]\vspace{-1mm}
  \[
    {\infer{}{\Lambda X.r\re\Lambda X.s}{r\re s}}
    \quad
    {\infer{}{\tapp{r}{A}\re \tapp{s}{A}}{r\re s}}
  \]
  \caption{Relations defining the operational semantics of PSI}
  \label{tab:SemOp}
\end{table}

%

\section{Examples}
\label{sec:ej1}

In this Section we present some examples to discuss uses and necessity for the
rules presented.

\newcommand{\apply}{\abs{f^{\timpl{A}{B}}}{\abs{x^A}{\app{f}{x}}}}

\begin{example}\label{ex:ex1}
We show the use of term equivalences to allow applications that are not possible
to build in System F. For instance, the ``apply'' function
$$\apply$$
can be applied to a pair, e.g. \pair{g}{r} with $\vdash g : \timpl{A}{B}$ and
$\vdash r : A$, because, due to isomorphism \eqref{iso:curry}, the type
derivation from Table~\ref{tab:ex1} is valid. Then we have
$$
  \app{(\apply)}{\pair{g}{r}}
  \rightleftarrows  \app{\app{(\apply)}{g}}{r}
  \hookrightarrow_{\beta_\lambda}^2  \app{g}{r}
$$
\begin{table*}
  \centering
  \infer{\imple}{\jud{}{\app{(\apply)}{\pair{g}{r}}} {B}}{
    \infer{\req}{\jud{}{\apply} {\timpl{(\tconj{(\timpl{A}{B})}{A})}{B}}}
    {\jud{}{\abs{f^{\timpl{A}{B}}}{\abs{x^A}{\app{f}{x}}}} {\timpl{(\timpl{A}{B})}{\timpl{A}{B}}}}
    \qquad
    \infer{\conji}{\jud{}{\pair{g}{r}} {\tconj{(\timpl{A}{B})}{A}}}
    {
      {\jud{}{g} {\timpl{A}{B}}}
      \qquad
      {\jud{}{r} {A}}
    }
  }
  \caption{Type derivation of example~\ref{ex:ex1}.}
  \label{tab:ex1}
\end{table*}
\end{example}

\begin{example}
Continuing with the previous example, equivalent applications can be build in other ways. For instance, the term
$$\app{\app{(\apply)}{r}}{g}$$
is well-typed using isomorphisms \eqref{iso:comm} and \eqref{iso:curry}, and reduces to \app{g}{r}:
\begin{align*}
  \app{\app{(\apply)}{r}}{g} &\rightleftarrows \app{(\apply)}{\pair{r}{g}}
  \\ &\rightleftarrows  \app{(\apply)}{\pair{g}{r}} \\ &\rightarrow^*  \app{g}{r}
\end{align*}  
\end{example}

\begin{example}
Concluding with the previous example, the uncurried ``apply'' function 
\[
  \abs{z^{\tconj{(\timpl{A}{B})}{A}}}{\app{\proj{\timpl{A}{B}}{z}}{\proj{A}{z}}}
\] can be applied to
$\vdash g : \timpl{A}{B}$ and $\vdash r : A$ as if it was curried:
\begin{align*}
  &\app{\app{(\abs{z^{\tconj{(\timpl{A}{B})}{A}}}{\app{\proj{\timpl{A}{B}}{z}}{\proj{A}{z}}})}{g}}{r}\\
  &\rightleftarrows\app{(\abs{z^{\tconj{(\timpl{A}{B})}{A}}}{\app{\proja{\timpl{A}{B}}{z}}{\proja{A}{z}}})}{\pair{g}{r}}\\
  &\hookrightarrow_{\beta_\lambda}\app{\proja{\timpl{A}{B}}{\pair{g}{r}}}{\proja{A}{\pair{g}{r}}}\\
  &\hookrightarrow_\pi^2 \app{g}{r}
\end{align*}
\end{example}
In the three previous examples, the $\beta$-reduction cannot occur
before the equivalences because of the typing condition in rule ($\beta_\lambda$).


\begin{example}
Another use of interest is the one mentioned in Section~\ref{sec:SystemI}: a
function returning a pair can be projected even while not being applied,
computing another function. Consider the term
$$\proja{\timpl{A}{B}}{\abs{x^A}{\pair{r}{s}}}$$
where $x:A\vdash r : B$ and $x:A\vdash s : C$.
This term is typable using isomorphism \eqref{iso:dist}, since $\timpl
A{(\tconj BC)}\equiv\tconj{(\timpl AB)}{(\timpl AC)}$.
The reduction goes as follows:
\begin{align*}
  {\proj{\timpl{A}{B}}{\abs{x^A}{\pair{r}{s}}}}
  &\rightleftarrows \proja{\timpl{A}{B}}{\pair{\abs{x^A}{r}}{\abs{x^A}{s}}} \\
  &\hookrightarrow_\pi \abs{x^A}{r}
\end{align*}
\end{example}

\begin{example}\label{ex:ex2}
Rule (\ref{spcommii}) is a consequence of isomorphism~\eqref{iso:pcommi}. For instance, the term $$\app{(\pAbs{X}{\abs{x^A}{\abs{f^{\timpl{A}{X}}}{\app{f}{x}}}})}{r}$$ is well-typed assuming $\vdash r:A$ and $X \not \in \ftva{A}$, as shown in~Table~\ref{tab:ex2},
and we have
\begin{align*}
  \app{(\pAbs{X}{\abs{x^A}{\abs{f^{\timpl{A}{X}}}{\app{f}{x}}}})}{r}
  &\rightleftarrows
  \app{(\abs{x^A}{(\pAbs{X}{\abs{f^{\timpl{A}{X}}}{\app{f}{x}}})})}{r}\\
  &\hookrightarrow_{\beta_\lambda} (\pAbs{X}{\abs{f^{\timpl{A}{X}}}{\app{f}{r}}})
\end{align*}
\begin{table*}
\centering
    \infer{\imple}{\jud{}{\app{(\pAbs{X}{\abs{x^A}{\abs{f^{\timpl{A}{X}}}{\app{f}{x}}}})}{r}} {\tfor{X}{(\timpl{(\timpl{A}{X})}{X}})}}
    {
      \infer{\req}{\jud{}{\pAbs{X}{\abs{x^A}{\abs{f^{\timpl{A}{X}}}{\app{f}{x}}}}} {\timpl{A}{\tfor{X}{(\timpl{(\timpl{A}{X})}{X}})}}}
      {\jud{}{\pAbs{X}{\abs{x^A}{\abs{f^{\timpl{A}{X}}}{\app{f}{x}}}}} {\tfor{X}{(\timpl{A}{\timpl{(\timpl{A}{X})}{X}})}}}
      \qquad
      \begin{array}{c}
	\\
	\jud{}{r}{A}
      \end{array}
    }
  \caption{Type derivation of example~\ref{ex:ex2}.}
  \label{tab:ex2}
\end{table*}
\end{example}

\begin{example}
Rule (\ref{spcommei}) is also a consequence of isomorphism~\eqref{iso:pcommi}. 
Consider the term
$$\app{\pApp{(\abs{x^{\tfor{X}{(\timpl{X}{X})}}}{x})}A}{\pAbs{X}{\abs{x^X}{x}}}$$
Let ${B}=\tfor X{(\timpl XX)}$. Since ${B}\Rightarrow{B}\equiv\forall Y.({B}\Rightarrow (Y\Rightarrow Y))$ (renaming the variable for readability), then
\[{\vdash{\app{\pApp{(\abs{x^{{B}}}{x})}A}{\pAbs{X}{\abs{x^X}{x}}}}:{\timpl{A}{A}}}\]
The reduction goes as follows:
\begin{align*}
  &\app{\pApp{(\abs{x^{\tfor{X}{(\timpl{X}{X})}}}{x})}A}{\pAbs{X}{\abs{x^X}{x}}} \\
  &\rightleftarrows
  \app{(\abs{x^{\tfor{X}{(\timpl{X}{X})}}}{\pApp{x}A})}{\pAbs{X}{\abs{x^X}{x}}} \\
  &\hookrightarrow_{\beta_\lambda}  \pApp{(\pAbs{X}{\abs{x^X}{x}})}A  \hookrightarrow_{\beta_\Lambda}  \abs{x^A}{x}
\end{align*}
\end{example}

\begin{example}
  Rules (\ref{spdistii}) and (\ref{spdistie}) are both consequences of the same
  isomorphism: \eqref{iso:pdist}. Consider the term
  \[
    \proj{\tfor{X}{(\timpl{X}{X})}}{\pAbs{X}{\pair{\abs{x^X}{x}}{r}}}
  \]
  where $\vdash r:A$. Since $\tfor X{(\tconj{(\timpl XX)}A)}\equiv\tconj{(\tfor
    X{(\timpl XX)})}{\tfor XA}$, we can derive
  $${\jud{}{\proj{\tfor{X}{(\timpl{X}{X})}}{\pAbs{X}{\pair{\abs{x^X}{x}}{r}}}}
    {\tfor{X}{(\timpl{X}{X})}}}$$
  A possible reduction is:
  \begin{align*}
    &\proj{\tfor{X}{(\timpl{X}{X})}}{\pAbs{X}{\pair{\abs{x^X}{x}}{r}}} \\ &\rightleftarrows \proja{\tfor{X}{(\timpl{X}{X})}}{\pair{\pAbs{X}{\abs{x^X}{x}}}{\pAbs{X}{r}}}\\
                                                                      &\hookrightarrow_\pi \pAbs{X}{\abs{x^X}{x}}
  \end{align*}
\end{example}

\begin{example}
Rule (\ref{spdistei}) is also a consequence of isomorphism \eqref{iso:pdist}.
  Consider
$$\pApp{\pair{\pAbs{X}{\abs{x^X}{\abs{y^A}{r}}}}{~\pAbs{X}{\abs{x^X}{\abs{z^B}{s}}}}}C$$
where $\vdash r:D$ and $\vdash s:E$.
    It has type $(C\Rightarrow A\Rightarrow
      D)\wedge(C\Rightarrow B\Rightarrow E)$, and
reduces as follows:
\begin{align*}
  &\pApp{\pair{\pAbs{X}{\abs{x^X}{\abs{y^A}{r}}}}{~\pAbs{X}{\abs{x^X}{\abs{z^B}{s}}}}}C \\
  &\rightleftarrows\pair{\pApp{(\abs{x^X}{\abs{y^A}{r}})}C}{\allowbreak\pApp{(\abs{x^X}{\abs{z^B}{s}})}C}\\
  &\hookrightarrow_{\beta_\Lambda}\pair{\abs{x^C}{\abs{y^A}{r}}}{~\abs{x^C}{\abs{z^B}{s}}}
\end{align*}
\end{example}

\begin{example}
Rule (\ref{spdistee}), too, is a consequence of isomorphism
\eqref{iso:pdist}. Consider the term
\[
  \pApp{(\proj{\tfor{X}{(\timpl{X}{X})}}{\pAbs{X}{\pair{\abs{x^X}{x}}{r}}})}A
\]
with type $A\Rightarrow A$, which reduces as follows:
\begin{align*}
  &\pApp{(\proj{\tfor{X}{(\timpl{X}{X})}}{\pAbs{X}{\pair{\abs{x^X}{x}}{r}}})}A \\ &\rightleftarrows \proj{\timpl{A}{A}}{\pApp{(\pAbs{X}{\pair{\abs{x^X}{x}}{r}})}{A}}  \\
  &\hookrightarrow_{\beta_\Lambda} \proja{\timpl{A}{A}}{\pair{\abs{x^A}{x}}{\substa{X}{A}{r}}} \\
  &\hookrightarrow_\pi \abs{x^A}{x}
\end{align*}
\end{example}
%
\section{Subject reduction}\label{sec:SR}
In this section we prove the preservation of typing through reduction. First we
need to characterise the equivalences between types, for example, if $\forall
X.A\equiv B\wedge C$, then $B\equiv\forall X.B'$ and $C\equiv\forall X.C'$, with
$A\equiv B'\wedge C'$ (Lemma~\ref{psi-equiv-tfor-tconj}). Due to the amount of
isomorphisms, this kind of lemmas are not trivial. To prove these relations, we
first define the multiset of prime factors of a type (Definition~\ref{def:PF}).
That is, the multiset of types that are not equivalent to a conjunction, such
that the conjunction of all its elements is equivalent to a certain type. This
technique has already been used in System I~\cite{DiazcaroDowekFSCD19}, however,
it has been used with simple types with only one basic type $\tau$. In PSI,
instead, we have an infinite amount of variables acting as basic types, hence
the proof becomes more complex.

We write $\vec X$ for $X_1,\dots, X_n$ and $\forall\vec X.A$ for $\forall
X_1.\dots\forall X_n.A$, for some $n$ (where, in the second case, if $n=0$,
$\forall\vec X.A=A$).
In addition, we write $[A_1,\dots,A_n]$ or $[A_i]_{i=1}^n$ for the multiset
containing the elements $A_1$ to $A_n$.


\begin{definition}[Prime factors]\label{def:PF}
  \begin{align*}
    \pf{ X } &= [X] \\
    \pf{ \timpl{A}{B} } &= [\tfor{\vec{X}_i}{(\timpl{(\tconj{A}{B_i})}{Y_i})}]_{i=1}^n\\
             &\qquad\qquad \text{ where } \pf{B} = [\tfor{\vec{X}_i}{(\timpl{B_i}{Y_i})}]_{i=1}^n \\
    \pf{ \tconj{A}{B} } &= \pf{A} \uplus \pf{B} \\
    \pf{ \tfor{X}{A} } &= [\tfor{X}{\tfor{\vec{Y}_i}{(\timpl{A_i}{Z_i})}}]_{i=1}^n\\
             &\qquad\qquad \text{ where } \pf{A} = [\tfor{\vec{Y}_i}{(\timpl{A_i}{Z_i})}]_{i=1}^n 
  \end{align*}  
\end{definition}

Lemma~\ref{lem:correctnessOne} and~Corollary~\ref{cor:correctness} state the correctness
of Definition~\ref{def:PF}. We write $\conj{[A_i]_{i=1}^n}$ for $\bigwedge_{i=1}^n A_i$.

\begin{lemma}\label{lem:correctnessOne}
  For all $A$, there exist $\vec X,n,B_1,\ldots,B_n,Y_1,\ldots,Y_n$ such that $PF(A)=[\tfor{\vec X_i}{(\timpl{B_i}{Y_i})}]_{i=1}^n$.
\end{lemma}
\begin{proof}
  Straightforward induction on the structure of $A$. 
\end{proof}

\begin{corollary}\label{cor:correctness}
  For all $A$, $A \equiv \conj{\pf{A}}$.
\end{corollary}
\begin{proof}
  By induction on the structure of A.

  \begin{itemize}
    \item Let $A = X$. Then $\pf{X} = [X]$, and \conj{[X]} $= X$.

    \item Let $A = \timpl{B}{C}$.
      By Lemma~\ref{lem:correctnessOne}, $\pf{C} =
      [\tfor{\vec{X}_i}{(\timpl{C_i}{Y_i})}]_{i=1}^n$.
      Hence, by definition, $\pf{A} =
      [\tfor{\vec{X}_i}{(\timpl{\tconj{B}{C_i}}{Y_i})}]_{i=1}^n$.
      By the induction hypothesis, $C \equiv \conj{\pf C} = \bigwedge_{i=1}^n \tfor{\vec{X}_i}{(\timpl{C_i}{Y_i})}$.
      Therefore,
      \begin{align*}
        A = \timpl{B}{C} 
	&\equiv \timpl{B}{\bigwedge_{i=1}^n \tfor{\vec{X}_i}{(\timpl{C_i}{Y_i})}}\\
	&\equiv \bigwedge_{i=1}^n \tfor{\vec{X}_i}{\allowbreak(\timpl{(\tconj{B}{C_i})}{Y_i})}\\
	&= \conj{[\tfor{\vec{X}_i}{(\timpl{\tconj{B}{C_i}}{Y_i})}]_{i=1}^n} = \conj{\pf{A}}
      \end{align*}
    \item Let $A = \tconj{B}{C}$.
      By the induction hypothesis, $B\equiv\conj{\pf B}$ and $C\equiv\conj{\pf
        C}$. Hence,
	\begin{align*}
        A 
	&=\tconj BC\equiv\conj{\pf B}\wedge\conj{\pf C}\\
	&\equiv \conj{\pf B\uplus\pf C} = \conj{\pf A}
	\end{align*}

    \item Let $A = \tfor{X}{B}$.
      By Lemma~\ref{lem:correctnessOne}, $\pf B=[\tfor{\vec
        Y_i}{(\timpl{B_i}{Z_i})}]_{i=1}^n$.
      Hence, by definition, $\pf{A} =
      [\tfor{X}{\tfor{\vec{Y}_i}{(\timpl{B_i}{Z_i})}}]_{i=1}^n$.
      By the induction hypothesis, $B \equiv \conj{\pf B} = \bigwedge_{i=1}^n
      \tfor{\vec{Y}_i}{(\timpl{B_i}{Z_i})}$.
      Therefore,
      \begin{align*}
        A &= \tfor{X}{B} \equiv \tfor{X}{\bigwedge_{i=1}^n  \tfor{\vec{Y}_i}{(\timpl{B_i}{Z_i})}}\\
	&\equiv \bigwedge_{i=1}^n \tfor{X}{\tfor{\vec{Y}_i}{\allowbreak(\timpl{B_i}{Z_i})}}\\
	&= \conj{[\tfor{X}{\tfor{\vec{Y}}{(\timpl{B_i}{Z})}}]_{i=1}^n} = \conj{\pf{A}}
	\end{align*}
	\qedhere
  \end{itemize}
\end{proof}

Lemma~\ref{lem:stability} states the stability of prime factors through
equivalence and Lemma~\ref{pf-sim} states a kind of reciprocal result. 
\begin{definition}
  $[A_1,\dots,A_n] \sim [B_1,\dots,B_m]$ if $n = m$ and $A_i \equiv B_{p(i)}$, for $i = 1,\dots,n$ and $p$ a permutation on $\{1,\dots.n\}$.
\end{definition}

\begin{lemma}\label{lem:stability}
  For all $A,B$ such that $A \equiv B$, we have $\pf{A} \sim \pf{B}$.
\end{lemma}
\begin{proof}
  First we check that \pf{\tconj{A}{B}} $\sim$ \pf{\tconj{B}{A}} and similar for the other five isomorphisms. Then we prove by structural induction that if $A$ and $B$ are equivalent in one step, then \pf{A} $\sim$ \pf{B}. We conclude by an induction on the length of the derivation of the equivalence $A \equiv B$.
\end{proof}

\begin{lemma}
  \label{pf-sim}
  For all $R,S$ multisets such that $R \sim S$, we have $\conj{R} \equiv \conj{S}$.
  \qed
\end{lemma}

\begin{lemma}\label{lem:eqprimes}
  For all $\vec X,\vec Z,A,B,Y,W$ such that $\forall\vec X.(A\Rightarrow Y)\equiv\forall\vec Z.(B\Rightarrow W)$, we have $\vec X=\vec Z$, $A\equiv B$, and $Y=W$.
\end{lemma}
\begin{proof}
  By simple inspection of the isomorphisms. 
\end{proof}

\begin{lemma}
  \label{psi-equiv-timpl-tconj}
  For all $A,B,C_1,C_2$ such that $\timpl{A}{B} \equiv \tconj{C_1}{C_2}$, there exist $B_1,B_2$ such that $C_1 \equiv \timpl{A}{B_1}$, $C_2\equiv \timpl{A}{B_2}$ and $B \equiv \tconj{B_1}{B_2}$.
\end{lemma}
\begin{proof} By Lemma \ref{lem:stability}, $\pf{\timpl AB}\sim
  \pf{\tconj{C_1}{C_2}}=\pf{C_1}\uplus \pf{C_2}$.

  By Lemma~\ref{lem:correctnessOne}, let $\pf B = [\tfor{\vec
    X_i}{(\timpl{D_i}{Z_i})}]_{i=1}^n$, $ \pf{C_1} = [\tfor{\vec
    Y_j}{(\timpl{E_j}{Z'_j})}]_{j=1}^k$, and $\pf{C_2} = [\tfor{\vec
    Y_j}{(\timpl{E_j}{Z'_j})}]_{j=k+1}^m$. Hence, $[{\tfor{\vec X_i}{(\timpl{(\tconj
        A{D_i})}{Z_i})}}]_{i=1}^n\sim[\tfor{\vec Y_j}{(\timpl{E_j}{Z'_j})}]_{j=1}^m$.
  So, by definition of $\sim$, $n=m$ and for $i=1,\dots,n$ and a permutation $p$,
  we have \( {\tfor{\vec X_i}{(\timpl{(\tconj A{D_i})}{Z_i})}} \equiv \tfor{\vec
    Y_{p(i)}}{(\timpl{E_{p(i)}}{Z'_{p(i)}})} \).
  Hence, by Lemma~\ref{lem:eqprimes}, we have $\vec X_i = \vec Y_{p(i)}$, $\tconj
  A{D_i}\equiv E_{p(i)}$, and $Z_i=Z'_{p(i)}$.

  Thus, there exists $I$ such that $I\cup \bar I=\{1,\dots,n\}$, such that
  \begin{align*} \pf{C_1} &=[{\tfor{\vec Y_{p(i)}}{(E_{p(i)}\Rightarrow
                            Z'_{p(i)})}}]_{i\in I}\\ \pf{C_2} &=[{\tfor{\vec Y_{p(i)}}{(E_{p(i)}\Rightarrow
                                                                Z'_{p(i)})}}]_{i\in\bar I}
  \end{align*} Therefore, by Corollary~\ref{cor:correctness},
  \[ C_1\equiv\bigwedge_{i\in I}\tfor{\vec Y_{p(i)}}{(E_{p(i)}\Rightarrow
      Z'_{p_i})} \equiv \bigwedge_{i\in I}{{\tfor{\vec X_i}{((\tconj
          A{D_i})\Rightarrow Z_i)}}}
  \] and
  \[
    C\equiv\bigwedge_{i\in\bar I}{{\tfor{\vec X_i}{((\tconj A{D_i})\Rightarrow Z_i)}}}
  \]
  Let $B_1={\bigwedge_{i\in I}{\tfor{\vec X_i}{(D_i\Rightarrow Z_i)}}}$ and $B_2
  ={\bigwedge_{i\in\bar I}{\tfor{\vec X_i}{(D_i\Rightarrow Z_i)}}}$. So,
  $C_1\equiv \timpl A{B_1}$ and $C_2\equiv\timpl A{B_2}$. In addition, also by
  Corollary~\ref{cor:correctness}, we have $B\equiv\bigwedge_{i=1}^n{\tfor{\vec
      X_i}{(D_i\Rightarrow Z_i)}}\equiv B_1\wedge B_2$.
\end{proof}

The proofs of the following two lemmas are similar to the proof of
Lemma~\ref{psi-equiv-timpl-tconj}. Full details are given in the technical appendix published at~\cite{arXiv}. 
\begin{lemma}
  \label{psi-equiv-tfor-tconj}
  For all $X,A,B,C$ such that $\tfor{X}{A} \equiv \tconj{B}{C}$, there exist
  $B',C'$ such that $B \equiv \tfor{X}{B'}$, $C \equiv \tfor{X}{C'}$ and $A \equiv
  \tconj{B'}{C'}$. \qed
\end{lemma}
\begin{lemma}
  \label{psi-equiv-tfor-timpl}
  For all $X,A,B,C$ such that $\tfor{X}{A} \equiv \timpl{B}{C}$, there exists
  $C'$ such that $C \equiv \tfor{X}{C'}$ and $A \equiv \timpl{B}{C'}$. \qed
\end{lemma}

Since the calculus is presented in Church-style, excluding rule $\req$, PSI is
syntax directed. Therefore, the generation lemma (Lemma~\ref{psi-gen}) is
straightforward, and we have the following unicity lemma (whose proof is given
in the technical appendix published at~\cite{arXiv}): 
\begin{lemma}
  [Unicity modulo]\label{lem:unicity}
  For all $\Gamma,r,A,B$ such that $\Gamma\vdash r:A$ and $\Gamma\vdash r:B$, we have $A\equiv B$.\qed
\end{lemma}

\begin{lemma}[Generation]
\label{psi-gen}
For all $\Gamma,x,r,s,X,A,B$:
\begin{enumerate}
	\item If $\hastypei{x}{A}$ and $\hastypei{x}{B}$, then $A \equiv B$.
	\item If $\hastypei{\abs{x^A}{r}}{B}$, then there exists $C$ such that $\hastypeg{\hastype{x}{A}}{r}{C}$ and $B \equiv A \Rightarrow C$.
	\item If $\hastypei{\app{ r }{ s }}{A}$, then there exists $C$ such that $\hastypei{r}{C \Rightarrow A}$ and $\hastypei{s}{C}$.
	\item If $\hastypei{\pair{ r }{ s }}{A}$, then there exist $C,D$ such that $A \equiv C \wedge D$, $\hastypei{r}{C}$ and $\hastypei{s}{D}$.
	\item If $\hastypei{\proja{ A }{ r }}{B}$, then $A \equiv B$ and there exists $C$ such that $\hastypei{r}{B} \wedge C$.
	\item If $\hastypei{\pAbs{ X }{ r }}{A}$, then there exists $C$ such that $A \equiv \forall X . C$, $\hastypei{r}{C}$ and $X \not \in \ftva{\Gamma}$.
	\item If $\hastypei{\pApp{ r }{ A }}{B}$, then there exists $C$ such that $\substa{X}{A}{C} \equiv B$ and $\hastypei{r}{\forall X . C}$.
    \qed
\end{enumerate}
\end{lemma}

The detailed proofs of Lemma~\ref{substitution} (Substitution) and
Theorem~\ref{subjectreduction} (Subject Reduction) are given in the technical appendix at~\cite{arXiv}.

\begin{lemma}[Substitution]
  \label{substitution}~
  \begin{enumerate}
  \item For all $\Gamma,x,r,s,A,B$ such that $\hastypeg{x:B}{r}{A}$ and $\hastypei{s}{B}$, we have $\hastypei{\substa{x}{s}{r}}{A}$.
  \item For all $\Gamma,r,X,A,B$ such that $\Gamma\vdash r:A$, we have $[X:=B]\Gamma\vdash [X:=B]r:[X:=B]A$.
    \qed
  \end{enumerate}
\end{lemma}

\begin{theorem}[Subject reduction]
\label{subjectreduction}
For all $\Gamma,r,s,A$ such that $\hastypei{r}{A}$ and $r \hookrightarrow s$ or $r \rightleftarrows s$, we have $\hastypei{s}{A}$.
\qed
\end{theorem}
%

\section{Strong Normalisation}\label{sec:SN}
In this section we prove the strong normalisation of the relation $\to$, that
is, every reduction sequence fired from a typed term eventually terminates. The
set of typed strongly normalising terms with respect to reduction~$\to$ is
written $\SN$. The size of the longest reduction issued from $t$ is written
$|t|$.

We extend to polymorphism the proof of System I~\cite{DiazcaroDowekFSCD19}. To
prove that every term is in $\SN$, we associate, as usual, a set $\interp A$ of
strongly normalising terms to each type $A$. A term $\vdash r:A$ is said to be
reducible when $r\in\interp A$. We then prove an adequacy theorem stating that
every well typed term is reducible.

The set $\interp{A_1\Rightarrow A_2\Rightarrow \cdots\Rightarrow A_n\Rightarrow
X}$ can be defined either as the set of terms $r$ such that for all
$s\in\interp{A_1}$, $rs\in\interp{A_2\Rightarrow\cdots\Rightarrow A_n\Rightarrow
X}$ or, equivalently, as the set of terms $r$ such that for all $
s_i\in\interp{A_i}$, $rs_1\dots s_n\in\interp X=\SN$. To prove that a term of
the form $\lambda x^A.t$ is reducible, we need to use the so-called CR3
property~\cite{Girard89}, in the first case, and the property that a term whose
all one-step reducts are in $\SN$ is in $\SN$, in the second. In PSI, an
introduction can be equivalent to an elimination e.g.~$\pair{rt}{st}\eq
\pair{r}{s}t$, hence, we cannot define a notion of neutral term and have an
equivalent to the CR3 property. Therefore, we use the second definition, and
since reduction depends on types, the set $\interp A$ is defined as a set of
typed terms.

Before proving the normalisation of PSI, we reformulate the proof of strong
normalisation of System F along these lines.

\subsection{Normalisation of System F}\label{sec:SNinSF}

\begin{definition}[Elimination context] 
  Consider an extension of the language where we introduce an extra symbol
  $\hole A$, called hole of type $A$. We define the set of elimination contexts
  with a hole $\hole A$ as the smallest set such that:
  \begin{itemize}
  \item $\hole A$ is an elimination context of type $A$,
  \item if $\kk{}{B\Rightarrow C}{A}$ is an elimination context of type $B
    \Rightarrow C$ with a hole of type $A$ and $r\in\SN$ is a term of type $B$, then
    $\kk{}{B\Rightarrow C}{A}r$ is an elimination context of type $C$ with a hole of
    type $A$,
  \item and if $\kk{}{\tfor{X}{B}}{A}$ is an elimination context of type
    $\tfor{X}{B}$ with a hole of type $A$, then $\kk{}{\tfor{X}{B}}{A}[C]$ is an
    elimination context of type $\substa{X}{C}{B}$ with a hole of type $A$.
  \end{itemize}
  We write $\ka{}{B}{A}{r}$ for $\subst{\hole{A}}{r}{\kk{}{B}{A}}$, where
  $\hole A$ is the hole of $\kk{}{B}{A}$. In particular, $r$ may be an elimination
  context.
\end{definition}

Notice that the shape of every context \kk{}{B}{A} is $\hole{A} \alpha_1 \dots
\alpha_n$, where each $\alpha_i$ is either a term or a type argument.

\begin{example}
  Let
  \begin{align*}
    \kk{}{{X}}{{X}}&=\hole{X}\\
    \kk{'}{{X}}{{{X}\Rightarrow{X}}}&=\ka{}{X}{X}{\hole{{X}\Rightarrow{X}}x}\\
    \kk{''}{{X}}{{\tfor{X}{\timpl{X}{X}}} }&= \ka{'}{X}{{X}\Rightarrow{X}}{\tapp{\hole{\tfor{X}{\timpl{X}{X}}}}{X}}\\ &=\ka{}{{X}}{{X}}{\tapp{\hole{\tfor{X}{\timpl{X}{X}}}}{X}x)}
  \end{align*}
  
  Then $\ka{''}{{X}}{{\tfor{X}{\timpl{X}{X}}}}{\Lambda X.\lambda y^{X}.y}=(\Lambda X.\lambda y^{X}.y)[X]x$.
\end{example}

\begin{definition}
  [Terms occurring in an elimination context]
  Let $\kk{}{B}{A}$ be an elimination context.
  The multiset of terms occurring in $\kk{}{B}{A}$ is defined as
  \begin{align*}
    \T(\hole A) &= \emptyset\\
    \T(\kk{}{{B\Rightarrow C}}{A}r) &= \T(\kk{}{{B\Rightarrow C}}{A}) \uplus \{r\}\\
    \T(\tapp{\kk{}{{\tfor{X}{B}}}{A}}{C}) &= \T(\kk{}{{\tfor{X}{B}}}{A})
  \end{align*}
  We write $|\kk{}{B}{A}|$ for $\sum_{i=1}^n|r_i|$ where $[r_1,\dots,r_n]=\T(\kk{}{B}{A})$.
\end{definition}

\begin{definition}[Reducibility]
  The set of reducible terms of type $A$ (notation $\interp A$) is defined as
the set of terms $r$ of type $A$ such that for any elimination context
$\kk{}{X}{A}$ where all the terms in $\T(\kk{}{X}{A})$ are in $\SN$, we have
$\ka{}{X}{A}{r}\in\SN$.
\end{definition}

\begin{lemma}
  \label{lem:CR1SF}
  For all $A$, $\interp A\subseteq\SN$.
\end{lemma}
\begin{proof}
  For all $A$, there exists an elimination context \kk{}{X}{A}, since variables are in $\SN$ and they can have any type. Hence, given that if $r\in\interp A$ then $\ka{}{X}{A}{r} \in \SN$, we have $r \in \SN$.
\end{proof}

\begin{lemma}[Adequacy of variables]
  \label{lem:varSF}
  For all $A$ and $x^A$, we have $x^A\in\interp A$.
\end{lemma}
\begin{proof}
  Let $\kk{}{X}{A} = \hole{A} \alpha_1 \dots \alpha_n$, where for all $i$ such
  that $\alpha_i$ is a term, we have $\alpha_i \in \SN$, then for all $x$,
  $\ka{}{X}{A}{x} \in \SN$.
\end{proof}

\begin{lemma}[Adequacy of application]\label{lem:appSF}
  For all $r, s, A, B$ such that $r \in \interp{A \Rightarrow B}$ and $s \in
  \interp{A}$, we have $r s \in \interp{B}$.
\end{lemma}
\begin{proof}
  We need to prove that for every elimination context $\kk{}{X}{B}$, we have
  $\ka{}{X}{B}{rs}\allowbreak\in\SN$. Since $s\in\interp A$, $\kk{'}{X}{{A\Rightarrow
      B}}=\ka{}{X}{B}{\hole{A\Rightarrow B}s} \in \SN$, and 
  $r\in\interp{A\Rightarrow B}$,
  we have $\ka{'}{X}{{A\Rightarrow B}}{r}=\ka{}{X}{B}{rs}\in\SN$.
\end{proof}

\begin{lemma}[Adequacy of abstraction]\label{lem:lambdaSF}
  For all $t, r, x, A, B$ such that $t \in \interp{A}$ and $\substa{x}{t}{r} \in
\interp{B}$, we have $\lambda x^A.r \in \interp{A \Rightarrow B}$. 
\end{lemma}
\begin{proof}
  We need to prove that for every elimination context $\kk{}{X}{{A\Rightarrow
      B}}$, we have $\ka{}{X}{{A\Rightarrow B}}{\lambda x^A.r}\in\SN$, that is that
  all its one step reducts are in $\SN$. By Lemma~\ref{lem:varSF}, $x\in\interp
  A$, so $r\in\interp B\subseteq\SN$. We conclude with an induction on $|r| +
  |\kk{}{X}{{A\Rightarrow B}}|$. 
\end{proof}

\begin{lemma}[Adequacy of type application]\label{lem:tappSF}
For all $r, X, A, B$ such that $r \in \interp{\tfor{X}{A}}$, we have $\tapp{r}{B} \in
\interp{\substa{X}{B}{A}}$.
\end{lemma}
\begin{proof}
  We need to prove that for every elimination context
  $\kk{}{Y}{{\substa{X}{B}{A}}}$ we have $\ka{}{Y}{\substa{X}{B}{A}}{\tapp{r}{B}}\in\SN$. Let
  $\kk{'}{Y}{{\tfor{X}{A}}}=\ka{}{Y}{\substa{X}{B}{A}}{\tapp{\hole{\tfor{X}{A}}}{B}}
  \in \SN$, and since $r\in\interp{\tfor{X}{A}}$,
  we have $\ka{'}{Y}{\tfor{X}{A}}{r} = \ka{}{Y}{\substa{X}{B}{A}}{\tapp{r}{B}}\in\SN$.
\end{proof}

\begin{lemma}[Adequacy of type abstraction]\label{lem:tlambdaSF}
For all $r$, $X$, $A$, $B$ such that $\substa{X}{B}{r} \in \interp{\substa{X}{B}{A}}$,
we have $\Lambda X . r \in \interp{\tfor{X}{A}}$.
\end{lemma}
\begin{proof}
  We need to prove that for every elimination context $\kk{}{Y}{{\tfor{X}{A}}}$,
  we have $\ka{}{Y}{{\tfor{X}{A}}}{\Lambda X.r}\in\SN$, that is that all its one
  step reducts are in $\SN$. Since $\substa{X}{B}{r}\in\interp A\subseteq\SN$ and
  every term in $\T(\kk{}{Y}{\tfor{X}{A}})$ is in $\SN$, then all its one step
  reducts are in $\SN$.
\end{proof}

\begin{definition}[Adequate substitution]
  A substitution $\theta$ is adequate for a context $\Gamma$ (notation
  $\theta\vDash\Gamma$) if for all $x:A \in \Gamma$, we have
  $\theta(x)\in\interp{A}$.
\end{definition}

\begin{theorem}[Adequacy]\label{thm:adequacySF}
  For all $\Gamma, r, A$, and substitution $\theta$ such that $\Gamma \vdash r:A$
  and $\theta\vDash\Gamma$, we have $\theta r\in\interp{A}$.
\end{theorem}
\begin{proof}
  By induction on $r$, using Lemmas~\ref{lem:varSF} to~\ref{lem:tlambdaSF}. 
\end{proof}

\begin{theorem}[Strong normalisation]
  For all $\Gamma, r, A$ such that $\Gamma \vdash r:A$, we have $r\in\SN$.
\end{theorem}
\begin{proof}
  By Lemma~\ref{lem:varSF}, the idendity substitution is adequate. Thus, 
  by Theorem~\ref{thm:adequacySF} and Lemma~\ref{lem:CR1SF}, 
  $r\in\interp{A}\subseteq\SN$.
\end{proof}

\subsection{Measure on terms}\label{sec:MT}

The size of a term is not invariant through the equivalence $\rightleftarrows$.
for example, counting the number of lambda abstractions in a term, we see that
$\lambda x^A.\pair rs$ is different than $\pair{\lambda x^A.r}{\lambda x^A.s}$.
Hence we introduce a measure $M(\cdot)$ on terms.

\begin{definition}
  [Measure on terms]
  \begin{align*}
    P(x)&= 0\\
    P(\lambda x^A.r)&= P(r)\\
    P(rs)&= P(r)\\
    P(\pair{r}{s})&= 1 + P(r) + P(s)\\
    P(\pi_A(r))&= P(r)\\
    P(\Lambda X.r)&= P(r)\\
    P(\tapp{r}{A})&= P(r)
  \end{align*}
  \begin{align*}
    M(x)&= 1\\ 
    M(\lambda x^A.r)&= 1 + M(r) + P(r)\\ 
    M(rs)&= M(r) + M(s) + P(r) M(s)\\ 
    M(\pair{r}{s})&= M(r) + M(s)\\ 
    M(\pi_A(r))&= 1 +  M(r) + P(r)\\
    M(\Lambda X.r)&= 1 + M(r) + P(r)\\ 
    M(\tapp{r}{A})&= 1 + M(r) + P(r)
  \end{align*}
\end{definition}

\begin{lemma}
  \label{lem:eqPr}
  For all $r, s$ such that $ r\eq s$, we have $P( r)=P( s)$.
\end{lemma}
\begin{proof}
  We check the case of each rule of Table~\ref{tab:SemOp}, and then conclude by structural induction to handle the contextual closure.
  \begin{itemize}
  \item \rulelabel{(comm)}: $\begin{aligned}[t]P( \pair{r}{s})&=1+P( r)+P( s)\\
      &=P( \pair{s}{r})\end{aligned}$
  \item \rulelabel{(asso)}: $\begin{aligned}[t] P(\pair{\pair{r}{s}}{t}) &=2+P( r)+P( s)+P( t) \\
      &=P( \pair{r}{\pair{s}{t}}) \end{aligned}$

  \item \rulelabel{(dist$_\lambda$)}: $\begin{aligned}[t]P(\lambda x^A.\pair{r}{s}) &=1+P( r)+P( s) \\
      &=P(\pair{\lambda x^A. r}{\lambda x^A. s})\end{aligned}$

  \item \rulelabel{(dist$_{\mathsf{app}}$)}: $\begin{aligned}[t]P(\pair{r}{s} t)&=1+P(r)+P(s) \\
      &=P( \pair{r t}{s t})\end{aligned}$

  \item \rulelabel{(curry)}: $\begin{aligned}[t]P(( r s) t) &=P( r) \\
      &=P( r\pair{s}{t})\end{aligned}$
    
  \item \rulelabel{(P-COMM$_{\forall_i\Rightarrow_i}$)}: $\begin{aligned}[t]P(\Lambda X . \lambda x^A . r) &= P(r) \\
      &= P(\lambda x^A . \Lambda X . r)\end{aligned}$
    
  \item \rulelabel{(P-COMM$_{\forall_e\Rightarrow_i}$)}: $\begin{aligned}[t]P(\pApp{(\abs{x^A}{r})}B) &= P(r) \\
      &= P(\abs{x^A}{\pApp{r}B})\end{aligned}$
    
  \item \rulelabel{(P-DIST$_{\forall_i\land_i}$)}: $\begin{aligned}[t]P(\Lambda X . \pair{ r }{ s }) &= 1 + P(r) + P(s) \\
      &= P(\pair{ \Lambda X . r }{ \Lambda X . s })\end{aligned}$
    
  \item \rulelabel{(P-DIST$_{\forall_e\land_i}$)}: $\begin{aligned}[t]P(\pair{ r }{ s }[A]) &= 1 + P(r) + P(s) \\
      &= P(\pair{ \tapp{r}{A} }{ \tapp{s}{A} })\end{aligned}$
    
  \item \rulelabel{(P-DIST$_{\forall_i\land_e}$)}: $\begin{aligned}[t]P(\pi_{\forall X . A}(\Lambda X . r)) &= P(r) \\
      &= P(\Lambda X . \pi_A (r))\end{aligned}$
    
  \item \rulelabel{(P-DIST$_{\forall_e\land_e}$)}: 
    $\begin{aligned}[t]P(\pApp{(\proja{\tfor{X}{B}}{r})}A) &= P(r) \\
      &= P(\proj{\substa{X}{A}{B}}{\pApp{r}A}) 
    \qedhere
    \end{aligned}$
  \end{itemize}
\end{proof}

\begin{lemma}
  \label{lem:eqM}
  For all $r, s$ such that $r \eq s$, we have $M(r) = M(s)$.
\end{lemma}
\begin{proof}
  We check the case of each rule of Table~\ref{tab:SemOp}, and then conclude by
  structural induction to handle the contextual closure.
  \begin{itemize}
  \item \rulelabel{(comm)}: $\begin{aligned}[t] & M( \pair{r}{s}) \\
      &= M( r) + M( s) \\
      &= M( \pair{s}{r})\end{aligned}$

  \item \rulelabel{(asso)}: $\begin{aligned}[t] & M(\pair{\pair{r}{s}}{t}) \\
      &= M( r) + M( s) + M( t) \\
      &= M( \pair{r}{\pair{s}{t}})\end{aligned}$

  \item \rulelabel{(dist$_{\lambda}$)}: $\begin{aligned}[t] & M(\lambda x^A . \pair{r}{s}) \\
      &= 2 + M( r) + M( s) + P( r) + P( s) \\
      &= M(\pair{\lambda x^A . r}{\lambda x^A . s}) \end{aligned}$

  \item \rulelabel{(dist$_{\mathsf{app}}$)}: 
    $\begin{aligned}[t] & M(\pair{r}{s} t) \\
      &=\begin{aligned}[t]
	&M( r) + M( s) + 2 M( t)\\
	&+ P( r) M( t) + P( s) M( t)
      \end{aligned}\\
    &= M( \pair{r t}{s t}) \end{aligned}$

  \item \rulelabel{(curry)}:
    $\begin{aligned}[t] & M(( r s) t) \\
      &=\begin{aligned}[t]
	&M( r)+M( s)+P( r)M( s)\\
	&+M( t)+P( r)M( t)
      \end{aligned}\\
      &= M( r\pair{s}{t}) \end{aligned}$

  \item \rulelabel{(p-comm$_{\forall_i\Rightarrow_i}$)}: $\begin{aligned}[t] &M(\Lambda X . \lambda x^A . r) \\
      &= 2 + M(r) + 2 P(r) \\
      &= M(\lambda x^A . \Lambda X . r)\end{aligned}$
    
  \item \rulelabel{(p-comm$_{\forall_e\Rightarrow_i}$)}: $\begin{aligned}[t] &M(\pApp{(\abs{x^A}{r})}B) \\
      &= 2 + M(r) + 2 P(r) \\
      &= M(\abs{x^A}{\pApp{r}B})\end{aligned}$
    
  \item \rulelabel{(p-dist$_{\forall_i\land_i}$)}: $\begin{aligned}[t] &M(\Lambda X . \pair{ r }{ s }) \\
      &= 1 + M(r) + M(s) \\
      &= M(\pair{ \Lambda X . r }{ \Lambda X . s })\end{aligned}$
    
  \item \rulelabel{(p-dist$_{\forall_e\land_i}$)}: $\begin{aligned}[t] &M(\pair{ r }{ s }[A]) \\
      &= 2 + M(r) + P(r) + M(s) + P(s) \\
      &= M(\pair{ \tapp{r}{A} }{ \tapp{s}{A} })\end{aligned}$
    
  \item \rulelabel{(p-dist$_{\forall_i\land_e}$)}: $\begin{aligned}[t] &M(\pi_{\forall X . A}(\Lambda X . r)) \\
      &= 2 + M(r) + 2 P(r) \\
      &= M(\Lambda X . \pi_A (r))\end{aligned}$
    
  \item \rulelabel{(p-dist$_{\land_e\forall_e}$)}: $\begin{aligned}[t]
      &M(\pApp{(\proja{\tfor{X}{B}}{r})}A) \\
      &= 2 + M(r) + 2 P(r) \\
      &= M(\proj{\substa{X}{A}{B}}{\pApp{r}A})
      \hspace{2.4cm}
      \qedhere
    \end{aligned}
    $
  \end{itemize}
\end{proof}

\begin{lemma}\label{lem:newlemma}
  For all $r, s, X, A$,
  \[
    \begin{array}{r@{\ }l@{\qquad}r@{\ }l}
      M(\lambda x^A.r)&> M(r)
      &M(\pair{r}{s}) &> M(s)\\
      M(rs) &> M(r)
      &M(\pi_A(r)) &> M(r)\\
      M(rs) &> M(s)
      &M(\Lambda X.r) &> M(r)\\
      M(\pair{r}{s}) &>  M(r)
      &M(\tapp{r}{A}) &> M(r)
    \end{array}
  \]
\end{lemma}
\begin{proof}
  For all $t$, $M(t)\geq 1$. We conclude by case inspection.
\end{proof}

\subsection{Reduction of a product}\label{sec:RP}

When typed lambda-calculus is extended with pairs, proving that if $r_1\in\SN$
and $r_2\in\SN$ then $ \pair{r_1}{r_2}\in\SN$ is easy. However, in System I and PSI
this property (Lemma~\ref{lem:prodOfSN}) is harder to prove, as it requires a
characterisation of the terms equivalent to the product $ \pair{r_1}{r_2}$
(Lemma~\ref{lem:eqProd}) and of all the reducts of this term
(Lemma~\ref{lem:reProd}).

\begin{lemma} \label{lem:eqProd}
  For all $r, s, t$ such that $\pair{r}{s}\eq^*t$, we have either
  \begin{enumerate}
  \item\label{it:eqProd-sum} $t=\pair{u}{v}$ where either
    \begin{enumerate}
    \item $u\eq^*\pair{t_{11}}{t_{21}}$ and $v\eq^*       \pair{t_{12}}{t_{22}}$ with $r\eq^*\pair{t_{11}}{t_{12}}$ and
      $s\eq^*\pair{t_{21}}{t_{22}}$, or
    \item $v\eq^*\pair{w}{s}$ with $r\eq^*\pair{u}{w}$, or
      any of the three symmetric cases, or
    \item $r\eq^*u$ and $s\eq^*v$, or the symmetric case.
    \end{enumerate}
  \item\label{it:eqProd-lam} $t=\lambda x^A.a$ and $a\eq^*\pair{a_1}{a_2}$ with $r\eq^*\lambda x^A.a_1$ and $s\eq^*\lambda x^A.a_2$.
  \item\label{it:eqProd-app} $t=av$ and $a\eq^*\pair{a_1}{a_2}$,
    with $r\eq^*a_1v$ and $s\eq^*a_2v$.
  \item\label{it:eqProd-tlam} $t=\Lambda X.a$ and $a\eq^*\pair{a_1}{a_2}$ with $r\eq^*\Lambda X.a_1$ and $s\eq^*\Lambda X.a_2$.
  \item\label{it:eqProd-tapp} $t=\tapp{a}{A}$ and $a\eq^*\pair{a_1}{a_2}$, with $r\eq^*a_1[A]$ and $s\eq^*a_2[A]$.
  \end{enumerate}
\end{lemma}
\begin{proof}
  By a double induction, first on $M(t)$ and then on the length of the relation
  $\eq^*$.
  Full details are given in the technical appendix at~\cite{arXiv}.
\end{proof}

\begin{lemma}
  \label{lem:reProd}
  For all $r_1, r_2, s, t$ such that $\pair{r_1}{r_2}\eq^* s\re t$, there exists
  $ u_1, u_2$ such that $ t\eq^*  \pair{u_1}{u_2}$ and either
  \begin{enumerate}
  \item $ r_1\toreq   u_1$ and $ r_2\toreq u_2$, 
  \item $ r_1\toreq u_1$ and $   r_2\eq^* u_2$, or
  \item $ r_1\eq^* u_1$ and $ r_2\toreq u_2$.
  \end{enumerate}
\end{lemma}
\begin{proof}
  By induction on $M(  \pair{r_1}{r_2})$.
  Full details are given in the technical appendix at~\cite{arXiv}.
\end{proof}

\begin{lemma}\label{lem:prodOfSN}
  For all $r_1, r_2$ such that $r_1\in\SN$ and $r_2\in\SN$, we have $\pair{r_1}{r_2}\in\SN$.
\end{lemma}
\begin{proof}
  By Lemma~\ref{lem:reProd}, from a reduction sequence starting from $
  \pair{r_1}{r_2}$ we can extract one starting from $r_1$, from $r_2$ or from
  both. Hence, this reduction sequence is finite.
\end{proof}

\subsection{Reducibility}\label{sec:Red}
\begin{definition}
  [Elimination context]
  Consider an extension of the language where we introduce an extra symbol
  $\hole A$, called hole of type $A$. We define the set of elimination contexts with a hole $\hole A$ as
  the smallest set such that:
  \begin{itemize}
  \item $\hole A$ is an elimination context of type $A$,
  \item if $\kk{}{B\Rightarrow C}{A}$ is an elimination context of type $B
    \Rightarrow C$ with a hole of type $A$ and $r\in\SN$ is a term of type $B$, then
    $\kk{}{B\Rightarrow C}{A}r$ is an elimination context of type $C$ with a hole
    of type $A$,
  \item if $\kk{}{B\wedge C}{A}$ is an elimination context of type $B \wedge
    C$ with a hole of type $A$, then $\pi_B(\kk{}{B\wedge C}{A})$ is an
    elimination context of type $B$ with a hole of type $A$.
  \item and if $\kk{}{\tfor{X}{B}}{A}$ is an elimination context of type
    $\tfor{X}{B}$ with a hole of type $A$, then $\kk{}{\tfor{X}{B}}{A}[C]$ is an
    elimination context of type $\substa{X}{C}{B}$ with a hole of type $A$.
  \end{itemize}
  We write $\ka{}{B}{A}{r}$ for $\subst{\hole{A}}{r}{\kk{}{B}{A}}$, where $\hole A$ is the hole of
  $\kk{}{B}{A}$. In particular, $r$ may be an elimination context.
\end{definition}
\begin{example}
  Let
  \begin{align*}
    \kk{}{{X}}{{X}}&=\hole{X},\\
    \kk{'}{{X}}{{{X}\Rightarrow}({X}\wedge{X})}&=\ka{}{{X}}{{X}}{\pi_{X}(\hole{{X}\Rightarrow({X}\wedge{X})}x)},\\
    \kk{''}{{X}}{{\tfor{X}{\timpl{X}{}(\tconj{X}{X})}}} &= \ka{'}{{X}}{{{X}\Rightarrow({X}\wedge{X})}}{\tapp{\hole{\tfor{X}{\timpl{X}{(\tconj{X}{X})}}}}{X}}\\ &=\ka{}{{X}}{{X}}{\pi_{X}(\tapp{\hole{\tfor{X}{\timpl{X}{(\tconj{X}{X})}}}}{X}x)}.
  \end{align*}
  
  Then,
  \[
    \ka{''}{{X}}{{\tfor{X}{\timpl{X}{(\tconj{X}{X})}}}}{\Lambda X.\lambda
      y^{X}.\pair{y}{y}}=\pi_{X}((\Lambda X.\lambda y^{X}.\pair{y}{y})[X]x).
  \]
\end{example}

\begin{definition}
  [Terms occurring in an elimination context]
  Let $\kk{}{A}{B}$ be an elimination context.
  The multiset of terms occurring in $\kk{}{A}{B}$ is defined as
  \begin{align*}
    \T(\hole A) &= \emptyset\\
    \T(\kk{}{{B\Rightarrow C}}{A}r) &= \T(\kk{}{{B\Rightarrow C}}{A}) \uplus \{r\}\\
    \T(\pi_B(\kk{}{{B\wedge C}}{A})) &=\T(\kk{}{{B\wedge C}}{A})\\
    \T(\tapp{\kk{}{{\tfor{X}{B}}}{A}}{C}) &= \T(\kk{}{{\tfor{X}{B}}}{A})
  \end{align*}
  We write $|\kk{}{B}{A}|$ for $\sum_{i=1}^n|r_i|$ where $[r_1,\dots,r_n]=\T(\kk{}{B}{A})$.
\end{definition}

\begin{example}
  We have that $\T(\hole Ars) = [r,s]$ and that $\T(\hole A\pair{r}{s}) = [\pair{r}{s}]$.
Remark that $\ka{}{B}{A}{t}\eq^* \ka{'}{B}{A}{t}$ does not imply $\T(\kk{}{B}{A})\sim{\mathcal
  T}(\kk{'}{B}{A})$.
\end{example}

\begin{definition}[Reducibility]\label{def:interpretation}
  The set of reducible terms of type $A$ (notation $\interp A$) is defined as
  the set of terms $r$ of type $A$ such that for any elimination context
  $\kk{}{X}{A}$ where all the terms in $\T(\kk{}{X}{A})$ are in $\SN$, we have
  $\ka{}{X}{A}{r}\in\SN$.
\end{definition}

The following lemma is a trivial consequence of the definition of reducibility.
\begin{lemma}\label{lem:eqInterp}
  For all $A, B$ such that $A\equiv B$, we have $\interp A=\interp B$.
  \qed
\end{lemma}

\begin{lemma}
  \label{lem:CR1}
  For all $A$, $\interp A\subseteq\SN$.
\end{lemma}
\begin{proof}
  For all $A$, there exists an elimination context \kk{}{X}{A}, since variables
  are in $\SN$ and they can have any type. Hence, given that if $r\in\interp A$ then
  $\ka{}{X}{A}{r} \in \SN$, we have $r \in \SN$.
\end{proof}

\subsection{Adequacy}\label{sec:Adequacy}
We finally prove the adequacy theorem (Theorem~\ref{thm:adequacy}) showing that
every typed term is reducible, and the strong normalisation theorem
(Theorem~\ref{thm:SN}) as a consequence of it.

\begin{lemma}[Adequacy of variables]
  \label{lem:var}
  For all $A$ and $x^A$, we have $x^A\in\interp A$.
\end{lemma}
\begin{proof}
  We need to prove that $\ka{}{X}{A}{x} \in \SN$. The term \ka{}{X}{A}{x} has
  the variable $x$ in a position that does not create any redex, hence the only
  redexes are those in $\T(\kk{}{X}{A})$, which are already in $\SN$. Then,
  $\ka{}{X}{A}{x} \in \SN$.
\end{proof}

\begin{lemma}[Adequacy of projection]\label{lem:SNimpliespiSN}
  For all $r, A, B$, such that $r\in\interp{A\wedge B}$, we have
  $\pi_A(r)\in\interp A$.
\end{lemma}
\begin{proof}
  We need to prove that $\ka{}{X}{A}{{\pi_A(r)}}\in\SN$. Take
  $\kk{'}{X}{{A\wedge B}}= \ka{}{X}{A}{{\pi_A(\hole{A\wedge B})}}$, and since
  $r\in\interp{A\wedge B}$,
  we have $\ka{'}{X}{A\wedge B}{{r}}=\ka{}{X}{A}{{\pi_A(r)}}\in\SN$.
\end{proof}

\begin{lemma}
  [Adequacy of application]\label{lem:appOfInterp}
  For all $r, s, A, B$ such that $r\in\interp{A\Rightarrow B}$ and $s\in\interp
  A$, we have $rs\in\interp B$.
\end{lemma}
\begin{proof}
  We need to prove that $\ka{}{X}{B}{rs}\in\SN$. Take $\kk{'}{X}{{A\Rightarrow
      B}}=\ka{}{X}{B}{\hole{A\Rightarrow B}s}$, and since $r\in\interp{A\Rightarrow
    B}$,
    we have $\ka{'}{X}{A\Rightarrow B}{{r}}=\ka{}{X}{B}{rs}\in\SN$.
\end{proof}

\begin{lemma}[Adequacy of type application]\label{lem:tapp}
  For all $r$, $X$, $A$, $B$ such that $r\in \interp{\tfor{X}{A}}$, we have
  $\tapp{r}{B} \in \interp{\substa{X}{B}{A}}$.
\end{lemma}
\begin{proof}
  We need to prove that $\ka{}{Y}{\substa{X}{B}{A}}{{\tapp{r}{B}}}\in\SN$.
  Take
  $\kk{'}{Y}{\tfor{X}{A}}=\ka{}{Y}{\substa{X}{B}{A}}{{\tapp{\hole{\tfor{X}{A}}}{B}}}
  \in \SN$, and since $r\in\interp{\tfor{X}{A}}$,
  we have $\ka{'}{Y}{\tfor{X}{A}}{{r}} = \ka{}{Y}{\substa{X}{B}{A}}{{\tapp{r}{B}}}\in\SN$.
\end{proof}

\begin{lemma}[Adequacy of product]\label{lem:adequacyOfProd}
  For all $r, s, A, B$ such that $r\in\interp A$ and $s\in\interp B$, we have $
  \pair{r}{s}\in\interp{A\wedge B}$.
\end{lemma}
\begin{proof}
  We need to prove that $\ka{}{X}{A\wedge B}{{\pair{r}{s}}}\in\SN$. We proceed
  by induction on the number of projections in $\kk{}{X}{{A\wedge B}}$. Since the
  hole of $\kk{}{X}{{A\wedge B}}$ has type $A\wedge B$, and $\ka{}{X}{A\wedge
    B}{{t}}$ has type $X$ for any $t$ of type $A\wedge B$, we can assume, without
  lost of generality, that the context $\kk{}{X}{{A\wedge B}}$ has the form
  $\ka{'}{X}{C}{\pi_C (\hole{A\wedge B}\alpha_1\dots \alpha_n)}$, where each
  $\alpha_i$ is either a term or a type argument. We prove that all
  $\ka{'}{X}{C}{\pi_C (\pair{r\alpha_1\dots \alpha_n}{s\alpha_1\dots
      \alpha_n})}\in\SN$ by showing, more generally, that if $r'$ and $s'$ are two
  reducts of $r\alpha_1\dots \alpha_n$ and $s\alpha_1\dots \alpha_n$, then
  $\ka{'}{X}{C}{\pi_C(\pair{r'}{s'})}\in\SN$. For this, we show that all its one
  step reducts are in $\SN$, by induction on $|\kk{'}{X}{C}| + |r'|+|s'|$. The
  full details are given in the technical appendix at~\cite{arXiv}.
\end{proof}

\begin{lemma}[Adequacy of abstraction]\label{lem:lamOfInterp}
  For all $t, r, x, A, B$ such that $t \in \interp{A}$ and $\substa{x}{t}{r} \in
  \interp{B}$, we have $\lambda x^A.r \in \interp{A \Rightarrow B}$. 
\end{lemma}
\begin{proof}
  By induction on $M(r)$.
  \begin{itemize}
  \item If $r\eq^* \pair{r_1}{r_2}$, then by Lemma~\ref{psi-gen}, we have
    $B\equiv B_1\wedge B_2$ with $r_1$ of type $B_1$ and $r_2$ of type $B_2$, and so
    by Lemma~\ref{substitution}, $\substa{x}{t}{r_1}$ has type $B_1$ and
    $\substa{x}{t}{r_2}$ has type $B_2$. Since $\substa{x}{t}{r}\in\interp B$, we
    have $\pair{\substa{x}{t}{r_1}}{\substa{x}{t}{r_2}}\in\interp{B}$. By
    Lemma~\ref{lem:SNimpliespiSN}, $\substa{x}{t}{r_1}\in\interp{B_1}$ and
    $\substa{x}{t}{r_2}\in\interp{B_2}$. By the induction hypothesis, $\lambda
    x^A.r_1\in\interp{A\Rightarrow B_1}$ and $\lambda x^A.r_2\in\interp{A\Rightarrow
      B_2}$, thus, by Lemma~\ref{lem:adequacyOfProd}, $\lambda x^A.r\eq^*\pair{\lambda
      x^A.r_1}{\lambda x^A.r_2}\in\interp{(A\Rightarrow B_1)\wedge(A\Rightarrow
      B_2)}$. Finally, by Lemma~\ref{lem:eqInterp}, we have $\interp{(A\Rightarrow
      B_1)\wedge(A\Rightarrow B_2)}=\interp{A\Rightarrow B}$.

  \item If $r\nneq^* \pair{r_1}{r_2}$, we need to prove that for any elimination
    context $\kk{}{X}{{A\Rightarrow B}}$, we have $\ka{}{X}{{A\Rightarrow
        B}}{\Lambda X. r}\in\SN$.
    Since $r$ and all the terms in $\T(\kk{}{X}{{A\Rightarrow B}})$ are in $\SN$,
    we proceed by induction on the lexicographical order of
    $(|\kk{}{X}{{A\Rightarrow B}}| + |r|,M(r))$ to show that all the one step
    reducts of $\ka{}{X}{{A\Rightarrow B}}{\lambda x^A.r}$ are in $\SN$. Since $r$
    is not a product, its only one step reducts are the following.
    \begin{itemize}
    \item A term where the reduction took place in one of the terms in
      $\T(\kk{}{X}{{A\Rightarrow B}})$ or in $r$, and so we apply the induction
      hypothesis.
    \item $\ka{'}{X}{B}{\substa{x}{s}{r}}$, with $\ka{}{X}{{A\Rightarrow
          B}}{\lambda x^A.r}=\kk{'}{X}{{B}}[(\lambda x^A.r) s]$. Since
      $\substa{x}{s}{r}\in\interp B$, we have $\ka{'}{X}{B}{ \substa{x}{s}{r}}\in\SN$.
    \item $\ka{'}{X}{\timpl{A}{B'}}{\lambda x^A. \substa{X}{C}{r'}}$, with $r$
      $\eq^*$-equivalent to $\Lambda X . r'$, $B\equiv\tfor{X}{B'}$, and
      $\ka{}{X}{\timpl{A}{B}}{\lambda x^A. \Lambda X. r'}$ is equal to
      $\ka{'}{X}{\timpl{A}{B'}}{(\lambda x^A. \Lambda X. r')[C]}$. Since
      $M(\substa{X}{C}{r'}) < M(\Lambda X. r')$, we apply the induction hypothesis.
      \qedhere
    \end{itemize}
  \end{itemize}
\end{proof}

\begin{lemma}[Adequacy of type abstraction]\label{lem:tlambda}
For all $r$, $X$, $A$, $B$ such that $\substa{X}{B}{r} \in \interp{\substa{X}{B}{A}}$,
we have $\Lambda X . r \in \interp{\tfor{X}{A}}$.
\end{lemma}
\begin{proof}
  We proceed by induction on $M(r)$ with a proof similar to that of
  Lemma~\ref{lem:lamOfInterp}. Full details are given in the technical appendix at~\cite{arXiv}.
\end{proof}

\begin{definition}[Adequate substitution]
  A substitution $\theta$ is adequate for a context $\Gamma$ (notation
  $\theta\vDash\Gamma$) if for all $x:A \in \Gamma$, we have
  $\theta(x)\in\interp{A}$.
\end{definition}

\begin{theorem}[Adequacy]\label{thm:adequacy}
  For all $\Gamma, r, A$, and substitution $\theta$ such that $\Gamma \vdash
  r:A$ and $\theta\vDash\Gamma$, we have $\theta r\in\interp{A}$.
\end{theorem}
\begin{proof}
  By induction on $r$. 
  \begin{itemize}
  \item If $r$ is a variable $x:A \in \Gamma$, then, since $\theta \vDash
    \Gamma$, we have $\theta r\in\interp A$.
    
  \item If $r$ is a product $\pair{s}{t}$, then by Lemma~\ref{psi-gen}, $\Gamma
    \vdash s:B$, $\Gamma \vdash t:C$, and $A\equiv B\wedge C$, thus, by the induction
    hypothesis, $\theta s\in\interp B$ and $\theta t\in\interp C$. By
    Lemma~\ref{lem:adequacyOfProd}, $\pair{\theta s}{\theta t}\in\interp{B\wedge
      C}$, hence by Lemma~\ref{lem:eqInterp}, $\theta r\in\interp A$.
    
  \item If $r$ is a projection $\pi_A(s)$, then by Lemma~\ref{psi-gen},
    $\Gamma \vdash s:A\wedge B$, and by the induction hypothesis, $\theta
    s\in\interp{A\wedge B}$. By Lemma~\ref{lem:SNimpliespiSN}, $\pi_A(\theta
    s)\in\interp{A}$, hence $\theta r\in\interp A$.
    
  \item If $r$ is an abstraction $\lambda x^B.s$, with $\Gamma \vdash s:C$, then
    by Lemma~\ref{psi-gen}, $A\equiv B\Rightarrow C$, hence by the induction
    hypothesis, for all $\theta$ and for all $t\in\interp B$, $\subst{x}{t}{\theta
      s}\in\interp C$. Hence, by Lemma~\ref{lem:lamOfInterp}, ${\lambda x^B.\theta
      s}\in\interp{B\Rightarrow C}$, so, by Lemma~\ref{lem:eqInterp}, $\theta
    r\in\interp A$.
    
  \item If $r$ is an application $st$, then by Lemma~\ref{psi-gen}, $\Gamma
    \vdash s:B\Rightarrow A$ and $\Gamma \vdash t:B$, thus, by the induction
    hypothesis, $\theta s\in\interp{B\Rightarrow A}$ and $\theta t\in\interp B$.
    Hence, by Lemma~\ref{lem:appOfInterp}, we have $\theta r=\theta s\theta
    t\in\interp A$.

  \item If $r$ is a type abstraction $\Lambda X.s$, with $\Gamma \vdash s:B$,
    then by Lemma~\ref{psi-gen}, $A\equiv \tfor{X}{B}$, hence by the induction
    hypothesis, for all $\theta$, $\theta s \in \interp B$. Hence, by Lemma~\ref{lem:tlambda},
    ${\Lambda X.\theta s}\in\interp{\tfor{X}{B}}$, hence, by
    Lemma~\ref{lem:eqInterp}, $\theta r\in\interp A$.
    
  \item If $r$ is a type application $\tapp{s}{C}$, then by Lemma~\ref{psi-gen},
    $\Gamma \vdash s:\tfor{X}{B}$ with $A \equiv \substa{X}{C}{B}$, thus, by the
    induction hypothesis, $\theta s\in\interp{\tfor{X}{B}}$. Hence, by
    Lemma~\ref{lem:tapp}, we have $\theta r=\theta \tapp{s}{C}\in\interp A$.
    \qedhere
  \end{itemize}
\end{proof}

\begin{theorem}[Strong normalisation]\label{thm:SN}
  For all $\Gamma, r, A$ such that $\Gamma \vdash r:A$, we have $r\in\SN$.
\end{theorem}
\begin{proof}
  By Lemma~\ref{lem:CR1}, the identity substitution is adequate. Thus, by
  Theorem~\ref{thm:adequacy} and Lemma~\ref{lem:CR1},
  $r\in\interp{A}\subseteq\SN$.
\end{proof}

\section{Conclusion, Discussion and Future Work}\label{sec:conclusiones}
System I is a simply-typed lambda calculus with pairs, extended with an
equational theory obtained from considering the type isomorphisms as equalities.
In this way, the system allows a programmer to focus on the meaning of programs,
ignoring the rigid syntax of terms within the safe context provided by type
isomorphisms. In this paper we have extended System I with polymorphism, and its
corresponding isomorphisms, enriching the language with a feature that most
programmers expect.

From a logical perspective, System I is a proof system for propositional logic,
where isomorphic propositions have the same proofs, and PSI extends System I
with the universal quantifier.

The main theorems in this paper prove subject reduction
(Theorem~\ref{subjectreduction}) and strong normalisation
(Theorem~\ref{thm:SN}). The proof of the latter is a non-trivial adaptation of
Girard's proof~\cite{Girard89} for System F.

\subsection{Swap}\label{sec:swap}
As mentioned in Section~\ref{sec:SystemI}, two isomorphisms for System F with
pairs, as defined by Di Cosmo \cite{DiCosmo95}, are not considered explicitly:
isomorphisms \eqref{iso:alpha} and \eqref{iso:swap}. However, the
isomorphism~\eqref{iso:alpha} is just the $\alpha$-equivalence, which has been
given implicitly, and so it has indeed been considered. The isomorphism that
actually was not considered is \eqref{iso:swap}, which allows to swap the type
abstractions:
\(
\forall X.\forall Y.A \equiv \forall Y.\forall X.A
\).
This isomorphism is the analogous to the isomorphism $A\Rightarrow B\Rightarrow
C\equiv B\Rightarrow A\Rightarrow C$ at the first order level, which is a
consequence of isomorphisms \eqref{iso:curry} and \eqref{iso:comm}. At this
first order level, the isomorphism induces the following equivalence:
\begin{align*}
  (\lambda x^A.\lambda y^B.r)st
  &\rightleftarrows(\lambda x^A.\lambda y^B.r)\langle s,t \rangle\\
  &\rightleftarrows(\lambda x^A.\lambda y^B.r)\langle t,s \rangle\\
  &\rightleftarrows(\lambda x^A.\lambda y^B.r)ts
\end{align*}

An alternative approach would have been to introduce an equivalence between
$\lambda x^A.\lambda y^B.r$ and $\lambda y^B.\lambda x^A.r$. However, in any
case, to keep subject reduction, the $\beta_\lambda$ reduction must verify that
the type of the argument matches the type of the variable before reducing. This
solution is not easily implementable for the $\beta_\Lambda$ reduction, since it
involves using the type as a labelling for the term and the variable, to
identify which term corresponds to which variable (leaving the posibility for
non-determinism if the ``labellings'' are duplicated), but at the level of types
we do not have a natural labelling.

Another alternative solution, in the same direction, is the one implemented by
the selective lambda calculus~\cite{GarrigueAitkaciPOPL94}, where only arrows,
and not conjunctions, were considered, and so only the ismorphism $A\Rightarrow
B\Rightarrow C\equiv B\Rightarrow A\Rightarrow C$ is treated. In the selective
lambda calculus the solution is indeed to include external labellings (not
types) to identify which argument is being used at each time. We could have
added a labelling to type applications, $t[A_X]$, together with the following
rule:
\(
  r[A_X][B_Y] \rightleftarrows r[B_Y][A_X]
\)
and so modifying the $\beta_\Lambda$ to
\(
  (\Lambda X.r)[A_X]\re [X:=A]r
\).

Despite that such a solution seems to work, we found that it does not contribute
to the language in any aspect, while it does make the system less readable.
Therefore, we have decided to exclude the isomorphism \eqref{iso:swap} for PSI.

Another remark is that while a rule such as $r[A] \rightleftarrows r[B]$ with $A \equiv B$ seems to be admissible in the system,  it may not be necessary. Indeed, the only rule that could benefit from it seems to be \eqref{spdistei} when used from right to left, which is only worthy when there is a $\Lambda$ that can be factorised, in which case the rule \eqref{spdistii} can be applied.

\subsection{Future work}
\subsubsection{Eta-expansion rule}
An extended fragment of an early version~\cite{DiazcaroDowekLSFA12} of System I
has been implemented~\cite{DiazcaroMartinezlopezIFL15} in Haskell. In such an 
implementation, we have added some ad-hoc rules in order to have a progression
property (that is, having only introductions as normal forms of closed terms).
For example, ``If $s$ has type $B$, then $(\lambda x^A.\lambda y^B.r)s\re\lambda
x^A.((\lambda y^B.r)s)$''. Such a rule, among others introduced in this
implementation, is a particular case of a more general $\eta$-expansion rule.
Certainly, with the rule $r\re \lambda x^A.rx$ we can derive
\begin{align*}
  (\lambda x^A.\lambda y^B.r)s
  &\re \lambda z^A. (\lambda x^A.\lambda y^B.r)sz\\
  &\rightleftarrows^* \lambda z^A. (\lambda x^A.\lambda y^B.r)zs\\
  &\re \lambda z^A.((\lambda y^B.\substa{x}{z}{r})s)
\end{align*}

In~\cite{DiazcaroDowek2020} we have showed that it is indeed the case that all
the ad-hoc rules from \cite{DiazcaroDowekLSFA12} can be lifted by adding
extensional rules.

In addition, the proof of the consistency of PSI as a language of proof-terms
for second-order logic has been intentionally left out of this paper. Indeed, as
shown in~\cite{DiazcaroDowekFSCD19}, it would require to restrict variables to
only have ``prime types'', that is non-conjunctive types. Such a restriction has
also been shown to be not necessary when the language is extended with eta
rules~\cite{DiazcaroDowek2020}. Therefore, we preferred to delay the proof of
consistency for a future version of PSI with $\eta$-rules.

\subsubsection{Implementation}
The mentioned implementation of an early version of System I, included a fix
point operator and numbers, showing some interesting programming examples. We
plan to extend such an implementation for polymorphism, following the design of
PSI.

\balance
\bibliographystyle{plain}
\bibliography{biblio}

\newpage
\appendix
\onecolumn

\section{Detailed proofs of Section~\ref{sec:SR}}\label{app:SR}
\recap{Lemma}{psi-equiv-tfor-tconj}
{For all $X,A,B,C$ such that $\tfor{X}{A} \equiv \tconj{B}{C}$, there exist $B',C'$ such that $B \equiv \tfor{X}{B'}$, $C \equiv \tfor{X}{C'}$ and $A \equiv \tconj{B'}{C'}$. }
\begin{proof}
By Lemma \ref{lem:stability}, $\pf{\tfor{X}{A}}\sim \pf{\tconj{B}{C}}=\pf{B}\uplus \pf{C}$.

By Lemma~\ref{lem:correctnessOne}, let
$\pf A  = [\tfor{\vec Y_i}{(\timpl{A_i}{Z_i})}]_{i=1}^n$, 
$\pf B  = [\tfor{\vec W_j}{(\timpl{D_j}{Z'_j})}]_{j=1}^k$, and
$\pf C  = [\tfor{\vec W_j}{(\timpl{D_j}{Z'_j})}]_{j=k+1}^m$.

Hence, $[\tfor X{\tfor{\vec Y_i}{(\timpl{A_i}{Z_i})}}]_{i=1}^n\sim[\tfor{\vec W_j}{(\timpl{D_j}{Z'_j})}]_{j=1}^m$.
So, by definition of $\sim$, $n=m$ and for $i=1,\dots,n$ and a permutation $p$,
we have
$\tfor X{\tfor{\vec Y_i}{(\timpl{A_i}{Z_i})}} \equiv \tfor{\vec W_{p(i)}}{(\timpl{D_{p(i)}}{Z'_{p(i)}})}$.
Thus, by Lemma~\ref{lem:eqprimes}, we have $X,\vec Y_i = \vec W_{p(i)}$,
$A_i\equiv D_{p(i)}$, and $Z_i=Z'_{p(i)}$.
Therefore, there exists $I$ such that $I\cup \bar I=\{1,\dots,n\}$, such that
$\pf B =[{\tfor{\vec W_{p(i)}}{(D_{p(i)}\Rightarrow Z'_{p(i)})}}]_{i\in I}$ and 
$\pf C =[{\tfor{\vec W_{p(i)}}{(D_{p(i)}\Rightarrow Z'_{p(i)})}}]_{i\in\bar I}$.
Hence, by Corollary~\ref{cor:correctness}, we have,
  $B\equiv\bigwedge_{i\in I}\tfor{\vec W_{p(i)}}{(D_{p(i)}\Rightarrow Z'_{p_i})}
\equiv
\bigwedge_{i\in I}\tfor X{{\tfor{\vec Y_i}{(A_i\Rightarrow Z_i)}}}$ and 
$C\equiv\bigwedge_{i\in\bar I}\tfor X{{\tfor{\vec Y_i}{(A_i\Rightarrow Z_i)}}}$.

  Let $B'={\bigwedge_{i\in I}{\tfor{\vec Y_i}{(A_i\Rightarrow Z_i)}}}$ and
  $C'={\bigwedge_{i\in\bar I}{\tfor{\vec Y_i}{(A_i\Rightarrow Z_i)}}}$. So,
  $B\equiv\forall X.B'$ and $C\equiv\forall X.C'$.
  Hence, also by Corollary~\ref{cor:correctness}, we have
  $A\equiv\bigwedge_{i=1}^n{\tfor{\vec Y_i}{(A_i\Rightarrow Z_i)}}\equiv B'\wedge C'$.
\end{proof}

\recap{Lemma}{psi-equiv-tfor-timpl}{
  For all $X,A,B,C$ such that $\tfor{X}{A} \equiv \timpl{B}{C}$, there exists $C'$ such that $C \equiv \tfor{X}{C'}$ and $A \equiv \timpl{B}{C'}$.
}
\begin{proof}
By Lemma \ref{lem:stability}, $\pf{\tfor{X}{A}}\sim \pf{\timpl{B}{C}}$.

By Lemma~\ref{lem:correctnessOne}, let
  $\pf A  = [\tfor{\vec Y_i}{(\timpl{A_i}{Z_i})}]_{i=1}^n$ and 
  $\pf C  = [\tfor{\vec W_j}{(\timpl{D_j}{Z'_j})}]_{j=1}^m$.
Hence, $[\tfor X{\tfor{\vec Y_i}{(\timpl{A_i}{Z_i})}}]_{i=1}^n\sim[\tfor{\vec
  W_j}{(\timpl{(B\wedge D_j)}{Z'_j})}]_{j=1}^m$.
So, by definition of $\sim$, $n=m$ and for $i=1,\dots,n$ and a permutation $p$,
we have
\(
  \tfor X{\tfor{\vec Y_i}{(\timpl{A_i}{Z_i})}} \equiv \tfor{\vec
    W_{p(i)}}{(\timpl{(B\wedge D_{p(i)})}{Z'_{p(i)}})}
 \)

Hence, by Lemma~\ref{lem:eqprimes}, we have $X,\vec Y_i = \vec W_{p(i)}$,
$A_i\equiv B\wedge D_{p(i)}$, and $Z_i=Z'_{p(i)}$.
Hence, by Corollary~\ref{cor:correctness},
\[
  C\equiv \bigwedge_{j=1}^n\tfor{\vec W_j}{(D_j\Rightarrow Z'_j)}
  \equiv\bigwedge_{i=1}^n\tfor{\vec W_{p(i)}}{(D_{p(i)}\Rightarrow Z'_{p(i)})}
  \equiv \bigwedge_{i=1}^n\tfor X{{\tfor{\vec Y_i}{(D_{p(i)}\Rightarrow Z_i)}}}
\]

Let
  $C'={\bigwedge_{i=1}^n{\tfor{\vec Y_i}{(D_{p(i)}\Rightarrow Z_i)}}}$. So,
  $C\equiv\forall X.C'$.
  
  Hence, also by Corollary~\ref{cor:correctness}, we have
  \[
  A\equiv\bigwedge_{i=1}^n{\tfor{\vec Y_i}{(A_i\Rightarrow Z_i)}}
  \equiv \bigwedge_{i=1}^n{\tfor{\vec Y_i}{((B\wedge D_{p(i)})\Rightarrow Z_i)}}
  \equiv B\Rightarrow\bigwedge_{i=1}^n{\tfor{\vec Y_i}{(D_{p(i)}\Rightarrow Z_i)}}
  \equiv \timpl B{C'}
  \qedhere
\]
\end{proof}

\xrecap{Lemma}{Unicity modulo}{lem:unicity}{
  For all $\Gamma,r,A,B$ such that $\Gamma\vdash r:A$ and $\Gamma\vdash r:B$, we have $A\equiv B$.
  }
\begin{proof}~
  \begin{itemize}
  \item If the last rule of the derivation of $\Gamma\vdash r:A$ is
    $(\equiv)$, then we have a shorter derivation of $\Gamma\vdash r:C$ with
    $C\equiv A$, and, by the induction hypothesis, $C\equiv B$, hence $A\equiv
    B$.
  \item If the last rule of the derivation of $\Gamma\vdash r:B$ is $(\equiv)$ we proceed in the same way.
  \item All the remaining cases are syntax directed.
\qedhere
  \end{itemize}
\end{proof}

\xrecap{Lemma}{Substitution}{substitution}{
  \begin{enumerate}
  \item For all $\Gamma,x,r,s,A,B$ such that $\hastypeg{x:B}{r}{A}$ and $\hastypei{s}{B}$, we have $\hastypei{\substa{x}{s}{r}}{A}$.
  \item For all $\Gamma,r,X,A,B$ such that $\Gamma\vdash r:A$, we have $[X:=B]\Gamma\vdash [X:=B]r:[X:=B]A$.
  \end{enumerate}
}
\begin{proof}~
  \begin{enumerate}
  \item 
    By structural induction on $r$.
    
    \begin{itemize}
    \item Let $r = x$. By Lemma \ref{psi-gen}, $A \equiv B$, thus \hastypei{s}{A}. Since \substa{x}{s}{x} $= s$, we have \hastypei{\substa{x}{s}{x}}{A}.
      
    \item Let $r = y$, with $y \not = x$. Since \substa{x}{s}{y} $= y$, we have \hastypei{\substa{x}{s}{y}}{A}.
      
    \item Let $r = \abs{x^C}{t}$. We have \subst{x}{s}{\abs{x^C}{t}} $=$ \abs{x^C}{t}, so \hastypei{\subst{x}{s}{\abs{x^C}{t}}}{A}.
      
    \item Let $r = \abs{y^C}{t}$, with $y \not = x$. By Lemma \ref{psi-gen}, $A \equiv \timpl{C}{D}$ and \hastypeg{\hastype{y}{C}}{t}{D}. By the induction hypothesis, \hastypeg{\hastype{y}{C}}{\substa{x}{s}{t}}{D}, and so, by rule \impli, \hastypei{\abs{y^C}{\substa{x}{s}{t}}}{\timpl{C}{D}}. Since \abs{y^C}{\substa{x}{s}{t}} $=$ \subst{x}{s}{\abs{y^C}{t}}, using rule \req, \hastypei{\subst{x}{s}{\abs{x^C}{t}}}{A}.
      
    \item Let $r = \app{t}{u}$. By Lemma \ref{psi-gen}, \hastypei{t}{\timpl{C}{A}} and \hastypei{u}{C}. By the induction hypothesis, \hastypei{\substa{x}{s}{t}}{\timpl{C}{A}} and \hastypei{\substa{x}{s}{u}}{C}, and so, by rule \imple, \hastypei{\app{(\substa{x}{s}{t})}{(\substa{x}{s}{u})}}{A}. Since \app{(\substa{x}{s}{t})}{(\substa{x}{s}{u})} $=$ \subst{x}{s}{\app{t}{u}}, we have \hastypei{\subst{x}{s}{\app{t}{u}}}{A}.
      
    \item Let $r = \pair{t}{u}$. By Lemma \ref{psi-gen}, \hastypei{t}{C} and \hastypei{u}{D}, with $A \equiv \tconj{C}{D}$. By the induction hypothesis, \hastypei{\substa{x}{s}{t}}{C} and \hastypei{\substa{x}{s}{u}}{D}, and so, by rule \conji, \hastypei{\pair{\substa{x}{s}{t}}{\substa{x}{s}{u}}}{\tconj{C}{D}}. Since \pair{\substa{x}{s}{t}}{\substa{x}{s}{u}} $=$ \substa{x}{s}{\pair{t}{u}}, using rule \req, we have \hastypei{\substa{x}{s}{\pair{t}{u}}}{A}.
      
    \item Let $r = \proja{A}{t}$. By Lemma \ref{psi-gen}, \hastypei{t}{\tconj{A}{C}}. By the induction hypothesis, \hastypei{\substa{x}{s}{t}}{\tconj{A}{C}}, and so, by rule \conje, \hastypei{\proj{A}{\substa{x}{s}{t}}}{A}. Since \proj{A}{\substa{x}{s}{t}} $=$ \subst{x}{s}{\proja{A}{t}}, we have \hastypei{\subst{x}{s}{\proja{A}{t}}}{A}.
      
    \item Let $r = \pAbs{X}{t}$. By Lemma \ref{psi-gen}, $A \equiv \tfor{X}{C}$ and \hastypei{t}{C}. By the induction hypothesis, \hastypei{\substa{x}{s}{t}}{C}, and so, by rule \fori, \hastypei{\pAbs{X}{\substa{x}{s}{t}}}{\tfor{X}{C}}. Since \pAbs{X}{\substa{x}{s}{t}} $=$ \subst{x}{s}{\pAbs{X}{t}}, using rule \req, we have \hastypei{\subst{x}{s}{\pAbs{X}{t}}}{A}.
      
    \item Let $r = \pApp{t}{C}$. By Lemma \ref{psi-gen}, $A \equiv \substa{X}{C}{D}$ and \hastypei{t}{\tfor{X}{D}}. By the induction hypothesis, \hastypei{\substa{x}{s}{t}}{\tfor{X}{D}}, and so, by rule \fore, \hastypei{\pApp{(\substa{x}{s}{t})}{C}}{\substa{X}{C}{D}}. Since \pApp{(\substa{x}{s}{t})}{C} $=$ \subst{x}{s}{\pApp{t}{C}}, using rule \req, we have \hastypei{\subst{x}{s}{\pApp{t}{C}}}{A}.
    \end{itemize}
  \item By induction on the typing relation.
    \newcommand{\hastypexb}[2]{\ensuremath{\substa{X}{B}{\Gamma} \vdash #1 : #2}}
    \newcommand{\hastypexbg}[3]{\ensuremath{\substa{X}{B}{\Gamma}, #1 \vdash #2 : #3}}
    \begin{itemize}
    \item \ax: Let \hastypeg{\hastype{x}{A}}{x}{A}. Then, using rule \ax, we have \hastypexbg{\hastype{x}{\substa{X}{B}{A}}}{\substa{X}{B}{x}}{\substa{X}{B}{A}}.
    \item \req: Let \hastypei{r}{A}, with $A \equiv C$. By the induction hypothesis, \hastypexb{\substa{X}{B}{r}}{\substa{X}{B}{C}}. Since $A \equiv C$, $\substa{X}{B}{A} \equiv \substa{X}{B}{C}$. Using rule \req, we have \hastypexb{\substa{X}{B}{r}}{\substa{X}{B}{A}}.
    \item \impli: Let \hastypei{\abs{x^C}{t}}{\timpl{C}{D}}. By the induction hypothesis, \hastypexbg{\hastype{x}{\substa{X}{B}{C}}}{\substa{X}{B}{t}}{\substa{X}{B}{D}}. Using rule \impli, \hastypexb{\abs{x^{\substa{X}{B}{C}}}{\substa{X}{B}{t}}}{\timpl{\substa{X}{B}{C}}{\substa{X}{B}{D}}}. Since \abs{x^{\substa{X}{B}{C}}}{\substa{X}{B}{t}} $=$ \subst{X}{B}{\abs{x^C}{t}}, we have \hastypexb{\subst{X}{B}{\abs{x^C}{t}}}{\subst{X}{B}{\timpl{C}{D}}}.
    \item \imple: Let \hastypei{\app{t}{s}}{D}. By the induction hypothesis,
      \hastypexb{\substa{X}{B}{t}}{\subst{X}{B}{\timpl{C}{D}}} and
      \hastypexb{\substa{X}{B}{s}}{\substa{X}{B}{C}}.
      Since \subst{X}{B}{\timpl{C}{D}} $=$ \timpl{\substa{X}{B}{C}}{\substa{X}{B}{D}}, using rule \imple, we have \hastypexb{\app{(\substa{X}{B}{t})}{(\substa{X}{B}{s})}}{\substa{X}{B}{D}}. Since \app{(\substa{X}{B}{t})}{(\substa{X}{B}{s})} $=$ \subst{X}{B}{\app{t}{s}}, we have \hastypexb{\subst{X}{B}{\app{t}{s}}}{\substa{X}{B}{D}}.
    \item \conji: Let \hastypei{\p{t}{s}}{\tconj{C}{D}}. By the induction hypothesis, \hastypexb{\substa{X}{B}{t}}{\substa{X}{B}{C}} and \hastypexb{\substa{X}{B}{s}}{\substa{X}{B}{D}}. Using rule \conji, \hastypexb{\p{\substa{X}{B}{t}}{\substa{X}{B}{s}}}{\tconj{\substa{X}{B}{C}}{\substa{X}{B}{D}}}. Since \p{\substa{X}{B}{t}}{\substa{X}{B}{s}} $=$ \substa{X}{B}{\p{t}{s}}, and \tconj{\substa{X}{B}{C}}{\substa{X}{B}{D}} $=$ \subst{X}{B}{\tconj{C}{D}}, we have \hastypexb{\substa{X}{B}{\p{t}{s}}}{\subst{X}{B}{\tconj{C}{D}}}.
    \item \conje: Let \hastypei{t}{\tconj{C}{D}}. By the induction hypothesis, \hastypexb{\substa{X}{B}{t}}{\subst{X}{B}{\tconj{C}{D}}}. Since \subst{X}{B}{\tconj{C}{D}} $=$ \tconj{\subst{X}{B}{C}}{\subst{X}{B}{D}}, using rule \conje we have \hastypexb{\proj{\substa{X}{B}{C}}{\substa{X}{B}{t}}}{\subst{X}{B}{C}}. Since \proja{\substa{X}{B}{C}}{\substa{X}{B}{t}} $=$ \substa{X}{B}{\proja{C}{t}}, we have \hastypexb{\substa{X}{B}{\proja{C}{t}}}{\subst{X}{B}{C}}.
    \item \fori: Let \hastypei{\pAbs{Y}{t}}{\tfor{Y}{C}}, with $X \not \in \ftva{\Gamma}$. By the induction hypothesis, \hastypexb{\substa{X}{B}{t}}{\substa{X}{B}{C}}. Since $X \not \in \ftva{\Gamma}$, $X \not \in \fvt{\substa{X}{B}{\Gamma}}$. Using rule \fori, we have \hastypexb{\pAbs{Y}{\substa{X}{B}{t}}}{\pAbs{Y}{\substa{X}{B}{C}}}. Since \pAbs{Y}{\substa{X}{B}{t}} $=$ \substa{X}{B}{\pAbs{Y}{t}}, and \tfor{Y}{\substa{X}{B}{C}} $=$ \substa{X}{B}{\tfor{Y}{C}}, we have \hastypexb{\substa{X}{B}{\pAbs{Y}{t}}}{\substa{X}{B}{\tfor{Y}{C}}}.
    \item \fore: Let \hastypei{\pApp{t}{D}}{\substa{Y}{D}{C}}. By the
      induction hypothesis,
      \hastypexb{\substa{X}{B}{t}}{\substa{X}{B}{\tfor{Y}{C}}}.
      Since \substa{X}{B}{\tfor{Y}{C}} $=$ \tfor{Y}{\substa{X}{B}{C}}, using
      rule \fore, we have
      \hastypexb{\pApp{(\substa{X}{B}{t})}{\substa{X}{B}{D}}}{\substa{Y}{\substa{X}{B}{D}}{{\substa{X}{B}{C}}}}.

      Since \pApp{(\substa{X}{B}{t})}{\substa{X}{B}{D}} $=$ \subst{X}{B}{\pApp{t}{D}}, and \substa{Y}{\substa{X}{B}{D}}{{\substa{X}{B}{C}}} $=$ \substa{X}{B}{\substa{Y}{D}{C}}, we have \hastypexb{\subst{X}{B}{\pApp{t}{D}}}{\substa{X}{B}{\substa{Y}{D}{C}}}. \qedhere
    \end{itemize}
  \end{enumerate}
\end{proof}

\xrecap{Theorem}{Subject reduction}{subjectreduction}{
  For all $\Gamma,r,s,A$ such that $\hastypei{r}{A}$ and $r \hookrightarrow s$ or $r \rightleftarrows s$, we have $\hastypei{s}{A}$.
}
\begin{proof}
  
  By induction on the rewrite relation.
  
  \begin{description}
    \item[\eqref{si:eq:comm}:]
      \pair{t}{r} $\rightleftarrows$ \pair{r}{t}
      \begin{description}
      \item[$(^{\rightarrow})$]~
        \deductionsb
        \deductiona
        {\hastypei{\p{t}{r}}{A}}
        {Hypothesis}
        \deductionc{A \equiv \tconj{B}{C}}{\hastypei{t}{B}}{\hastypei{r}{C}}{1, Lemma \ref{psi-gen}}
        \deductiona{\tconj{B}{C} \equiv \tconj{C}{B}}{Iso. \eqref{iso:comm}}
      \item \hfill
        \begin{center}
          \begin{prooftree}
            \hypo{\judi{r}{C}}
            \hypo{\judi{t}{B}}
            \infer2[\conji]{\judi{\p{r}{t}} {\tconj{C}{B}}}
            \infer[left label=[3]]1[\req]{\judi{\p{r}{t}} {\tconj{B}{C}}}
            \infer[left label=[2]]1[\req]{\judi{\p{r}{t}} {A}}
          \end{prooftree}
        \end{center} 
        \deductionse
        
      \item[$(_{\leftarrow})$] analogous to $(^{\rightarrow})$.
      \end{description}
      
    \item[\eqref{si:eq:asso}:]
      \pair{t}{\p{r}{s}} $\rightleftarrows$ \p{\pair{t}{r}}{s}
      \begin{description}
      \item [$(^{\rightarrow})$]~
        \deductionsb
        \deductiona
        {\hastypei{\pair{t}{\p{r}{s}}}{A}}
        {Hypothesis}
        \deductionc{A \equiv \tconj{B}{C}}{\hastypei{t}{B}}{\hastypei{\p{r}{s}}{C}}{1, Lemma \ref{psi-gen}}
        \deductionc{C \equiv \tconj{D}{E}}{\hastypei{r}{D}}{\hastypei{s}{E}}{2, Lemma \ref{psi-gen}}
        \deductiona{\tconj{B}{(\tconj{D}{E})} \equiv \tconj{(\tconj{B}{D})}{E}}{Iso. \eqref{iso:asso}}
        \deductiona{A \equiv \tconj{B}{(\tconj{D}{E})}}{2, 3, congr. \req}
      \item \hfill
        \begin{center}
          \begin{prooftree}
            \hypo{\judi{t}{B}}
            \hypo{\judi{r}{D}}
            \infer2[\conji]{\judi{\p{t}{r}} {\tconj{B}{D}}}
            \hypo{\judi{s}{E}}
            \infer2[\conji]{\judi{\p{\p{t}{r}}{s}} {\tconj{(\tconj{B}{D})}{E}}}
            \infer[left label=[4]]1[\req]{\judi{\p{\p{t}{r}}{s}} {\tconj{B}{(\tconj{D}{E})}}}
            \infer[left label=[5]]1[\req]{\judi{\p{\p{t}{r}}{s}} {A}}
          \end{prooftree}
        \end{center} 
        \deductionse
        
      \item [$(_{\leftarrow})$] analogous to $(^{\rightarrow})$.
      \end{description}
      
    \item[\eqref{si:eq:dist-abs}:]
      \abs{x^A}{\p{t}{r}} $\rightleftarrows$ \p{\abs{x^A}{t}}{\abs{x^A}{r}}
      \begin{description}
      \item [$(^{\rightarrow})$]~
        \deductionsb
        \deductiona
        {\hastypei{\abs{x^A}{\p{t}{r}}}{B}}
        {Hypothesis}
        \deductionb{B \equiv \timpl{A}{C}}{\hastypeg{\hastype{x}{A}}{\p{t}{r}}{C}}{1, Lemma \ref{psi-gen}}
        \deductionc{C \equiv \tconj{D}{E}}{\hastypeg{\hastype{x}{A}}{t}{D}}{\hastypeg{\hastype{x}{A}}{r}{E}}{2, Lemma \ref{psi-gen}}
        \deductiona{\timpl{A}{(\tconj{D}{E})} \equiv \tconj{(\timpl{A}{D})}{(\timpl{A}{E})}}{Iso. \eqref{iso:dist}}
        \deductiona{B \equiv \timpl{A}{(\tconj{D}{E})}}{2, 3, congr. \req}
      \item ~

	\scalebox{0.89}{\centering
          \begin{prooftree}
            \hypo{\judg{\hastype{x}{A}}{t}{D}}
            \infer1[\impli]{\judi{\abs{x^A}{t}} {\timpl{A}{D}}}
            \hypo{\judg{\hastype{x}{A}}{r}{E}}
            \infer1[\impli]{\judi{\abs{x^A}{r}} {\timpl{A}{E}}}
            \infer2[\conji]{\judi{\p{\abs{x^A}{t}}{\abs{x^A}{r}}} {\tconj{(\timpl{A}{D})}{(\timpl{A}{E})}}}
            \infer[left label=[4]]1[\req]{\judi{\p{\abs{x^A}{t}}{\abs{x^A}{r}}} {\timpl{A}{(\tconj{D}{E})}}}
            \infer[left label=[5]]1[\req]{\judi{\p{\abs{x^A}{t}}{\abs{x^A}{r}}} {B}}
          \end{prooftree}
      }
        \deductionse
        
      \item [$(_{\leftarrow})$]~
        \deductionsb
        \deductiona
        {\hastypei{\p{\abs{x^A}{t}}{\abs{x^A}{r}}}{B}}
        {Hypothesis}
        \deductionc{B \equiv \tconj{C}{D}}{\hastypei{\abs{x^A}{t}}{C}}{\hastypei{\abs{x^A}{r}}{D}}{1, Lemma \ref{psi-gen}}
        \deductionb{C \equiv \timpl{A}{C'}}{\hastypeg{\hastype{x}{A}}{t}{C'}}{2, Lemma \ref{psi-gen}}
        \deductionb{D \equiv \timpl{A}{D'}}{\hastypeg{\hastype{x}{A}}{r}{D'}}{2, Lemma \ref{psi-gen}}
        \deductiona{\tconj{(\timpl{A}{C'})}{(\timpl{A}{D'})} \equiv \timpl{A}{(\tconj{C'}{D'})}}{Iso. \eqref{iso:dist}}
        \deductiona{B \equiv \tconj{(\timpl{A}{C'})}{(\timpl{A}{D'})}}{2, 3, 4, congr. \req}
      \item \hfill
        \begin{center}
          \begin{prooftree}
            \hypo{\judg{\hastype{x}{A}}{t}{C'}}
            \hypo{\judg{\hastype{x}{A}}{r}{D'}}
            \infer2[\conji]{\judg{\hastype{x}{A}}{\p{t}{r}} {\tconj{C'}{D'}}}
            \infer1[\impli]{\judi{\abs{x^A}{\p{t}{r}}} {\timpl{A}{(\tconj{C'}{D'})}}}
            \infer[left label=[5]]1[\req]{\judi{\abs{x^A}{\p{t}{r}}} {\tconj{(\timpl{A}{C'})}{(\timpl{A}{D'})}}}
            \infer[left label=[6]]1[\req]{\judi{\abs{x^A}{\p{t}{r}}} {B}}
          \end{prooftree}
        \end{center} 
        \deductionse
      \end{description}
      
    \item[\eqref{si:eq:dist-app}:]
      \app{\p{t}{r}}{s} $\rightleftarrows$ \p{\app{t}{s}}{\app{r}{s}}
      \begin{description}
      \item [$(^{\rightarrow})$]~
        \deductionsb
        \deductiona
        {\hastypei{\app{\p{t}{r}}{s}}{A}}
        {Hypothesis}
        \deductionb{\hastypei{\p{t}{r}}{\timpl{B}{A}}}{\hastypei{s}{B}}{1, Lemma \ref{psi-gen}}
        \deductionc{\timpl{B}{A} \equiv \tconj{C}{D}}{\hastypei{t}{C}}{\hastypei{r}{D}}{2, Lemma \ref{psi-gen}}
        \deductionc{C \equiv \timpl{B}{C'}}{D \equiv \timpl{B}{D'}}{A \equiv \tconj{C'}{D'}}{3, Lemma \ref{psi-equiv-timpl-tconj}}
      \item~
        \begin{center}
          \begin{prooftree}
            \hypo{\judi{t}{C}}
            \infer[left label=[4]]1[\req]{\judi{t} {\timpl{B}{C'}}}
            \hypo{\judi{s}{B}}
            \infer2[\imple]{\judi{\app{t}{s}}{C'}}
          \end{prooftree}
        \end{center}
        \needspace{3em}
      \item~
        \begin{center}
          \begin{prooftree}
            \hypo{\judi{r}{D}}
            \infer[left label=[4]]1[\req]{\judi{r} {\timpl{B}{D'}}}
            \hypo{\judi{s}{B}}
            \infer2[\imple]{\judi{\app{r}{s}}{D'}}
          \end{prooftree}
        \end{center} 
      \item~ 
        \begin{center}
          \begin{prooftree}
            \hypo{$(5)$}
            \infer1{\judi{\app{t}{s}}{C'}}
            \hypo{$(6)$}
            \infer1{\judi{\app{r}{s}}{D'}}
            \infer2[\conji]{\judi{\p{\app{t}{s}}{\app{r}{s}}} {\tconj{C'}{D'}}}
            \infer[left label=[4]]1[\req]{\judi{\p{\app{t}{s}}{\app{r}{s}}} {A}}
          \end{prooftree}
        \end{center} 
        \deductionse
        
      \item [$(_{\leftarrow})$]~
        \deductionsb
        \deductiona
        {\hastypei{\p{\app{t}{s}}{\app{r}{s}}}{A}}
        {Hypothesis}
        \deductionc{A \equiv \tconj{B}{C}}{\hastypei{\app{t}{s}}{B}}{\hastypei{\app{r}{s}}{C}}{1, Lemma \ref{psi-gen}}
        \deductionb{\hastypei{t}{\timpl{D}{B}}}{\hastypei{s}{D}}{2, Lemma \ref{psi-gen}}
        \deductionb{\hastypei{r}{\timpl{E}{B}}}{\hastypei{s}{E}}{2, Lemma \ref{psi-gen}}
        \deductiona{D \equiv E}{3, 4, Lemma \ref{lem:unicity}}
        \deductiona{\timpl{D}{(\tconj{B}{C})} \equiv \tconj{(\timpl{D}{B})}{(\timpl{D}{C})}}{Iso. \eqref{iso:dist}}
        \deductiona{\timpl{E}{C} \equiv \timpl{D}{C}}{6, congr. \req}
      \item \hfill
        \begin{center}
          \begin{prooftree}
            \hypo{\judi{t}{\timpl{D}{B}}}
            \hypo{\judi{r}{\timpl{E}{C}}}
            \infer[left label=[7]]1[\req]{\judi{r}{\timpl{D}{C}}}
            \infer2[\conji]{\judi{\p{t}{r}} {\tconj{(\timpl{D}{B})}{(\timpl{D}{C})}}}
            \infer[left label=[5]]1[\req]{\judi{\p{t}{r}} {\timpl{D}{(\tconj{B}{C})}}}
            \infer1[\imple]{\judi{\app{\p{t}{r}}{s}} {\tconj{B}{C}}}
            \infer[left label=[2]]1[\req]{\judi{\app{\p{t}{r}}{s}} {A}}
          \end{prooftree}
        \end{center} 
        \deductionse
      \end{description}
      
    \item[\eqref{si:eq:curry}:]
      $\app{t}{\p{r}{s}} \rightleftarrows \app {\app{t}{r}{s}}$
      \begin{description}
      \item [$(^{\rightarrow})$]~
        \deductionsb
        \deductiona
        {\hastypei{\app{t}{\p{r}{s}}}{A}}
        {Hypothesis}
        \deductionb{\hastypei{t}{\timpl{B}{A}}}{\hastypei{\p{t}{r}}{B}}{1, Lemma \ref{psi-gen}}
        \deductionc{B \equiv \tconj{C}{D}}{\hastypei{r}{C}}{\hastypei{s}{D}}{2, Lemma \ref{psi-gen}}
        \deductiona{\timpl{B}{A} \equiv \timpl{(\tconj{C}{D})}{A}}{3, congr. \req}
        \deductiona{\timpl{(\tconj{C}{D})}{A} \equiv \timpl{C}{(\timpl{D}{A})}}{Iso. \eqref{iso:curry}}
      \item \hfill
        \begin{center}
          \begin{prooftree}
            \hypo{\judi{t}{\timpl{B}{A}}}
            \infer[left label=[4]]1[\req]{\judi{t} {\timpl{(\tconj{C}{D})}{A}}}
            \infer[left label=[5]]1[\req]{\judi{t} {\timpl{C}{(\timpl{D}{A})}}}
            \hypo{\judi{r}{C}}
            \infer2[\imple]{\judi{\app{t}{r}}{\timpl{D}{A}}}
          \end{prooftree}
        \end{center} \hfill
      \item \hfill
        \begin{center}
          \begin{prooftree}
            \hypo{$(6)$}
            \infer1{\judi{\app{t}{r}}{\timpl{D}{A}}}
            \hypo{\judi{s}{D}}
            \infer2[\imple]{\judi{\app{\app{t}{r}}{s}}{A}}
          \end{prooftree}
        \end{center} 
        \deductionse
        
      \item [$(_{\leftarrow})$]~
        \deductionsb
        \deductiona
        {\hastypei{\app{\app{t}{r}}{s}}{A}}
        {Hypothesis}
        \deductionb{\hastypei{\app{t}{r}}{\timpl{B}{A}}}{\hastypei{s}{B}}{1, Lemma \ref{psi-gen}}
        \deductionb{\hastypei{t}{\timpl{C}{(\timpl{B}{A})}}}{\hastypei{r}{C}}{2, Lemma \ref{psi-gen}}
        \deductiona{\timpl{C}{(\timpl{B}{A})} \equiv \timpl{(\tconj{C}{B})}{A}}{Iso. \eqref{iso:curry}}
      \item ~ 

        \scalebox{0.82}{\centering
            \begin{prooftree}
              \hypo{\judi{t}{\timpl{C}{(\timpl{B}{A})}}}
              \infer[left label=[4]]1[\req]{\judi{t}{\timpl{(\tconj{C}{B})}{A}}}
              \hypo{\judi{r}{C}}
              \hypo{\judi{s}{B}}
              \infer2[\conji]{\judi{\p{r}{s}} {\tconj{C}{B}}}
              \infer2[\imple]{\judi{\app{t}{\p{r}{s}}} {A}}
            \end{prooftree}}
        \deductionse
      \end{description}
      
    \item[\eqref{spcommii}:]
      $\pAbs{X}{\abs{x^A}{t}}\rightleftarrows\abs{x^A}{\pAbs{X}{t}}$
      \begin{description}
      \item[$(^{\rightarrow})$]~
        \deductionsb
        \deductiona
        {X \not \in \ftva{A}}
        {Hypothesis}
        \deductiona
        {\hastypei{\pAbs{X}{\abs{x^A}{t}}} {B}}
        {Hypothesis}
        \deductionc
        {B \equiv \tfor{X}{C}}
        {\hastypei{\abs{x^A}{t}} {C}}
        {X \not \in \ftva{\Gamma}}
        {2, Lemma \ref{psi-gen}}
        \deductionb
        {C \equiv \timpl{A}{D}}
        {\hastypeg{\hastype{x}{A}}{t}{D}}
        {3, Lemma \ref{psi-gen}}
        \deductiona
        {\tfor{X}{(\timpl{A}{D})} \equiv \timpl{A}{\tfor{X}{D}}}
        {1, Iso. \eqref{iso:pcommi}}
        \deductiona
        {\tfor{X}{C} \equiv \tfor{X}{(\timpl{A}{D})}}
        {4, congr. \req}
      \item \hfill
        \begin{center}
          \begin{prooftree}
            \hypo{\judg{\hastype{x}{A}}{t}{D}}
            \infer[left label=[1 3]]1[\fori]{\judg{\hastype{x}{A}}{\pAbs{X}{t}} {\tfor{X}{D}}}
            \infer1[\impli]{\judi{\abs{x^A}{\pAbs{X}{t}}} {\timpl{A}{\tfor{X}{D}}}}
            \infer[left label=[5]]1[\req]{\judi{\abs{x^A}{\pAbs{X}{t}}} {\tfor{X}{(\timpl{A}{D})}}}
            \infer[left label=[6]]1[\req]{\judi{\abs{x^A}{\pAbs{X}{t}}} {\tfor{X}{C}}}
            \infer[left label=[3]]1[\req]{\judi{\abs{x^A}{\pAbs{X}{t}}} {B}}
          \end{prooftree}
        \end{center} 
        \deductionse
        
      \item[$(_\leftarrow)$]~
        \deductionsb
        \deductiona
        {X \not \in \ftva{A}}
        {Hypothesis}
        \deductiona
        {\hastypei{\abs{x^A}{\pAbs{X}{t}}} {B}}
        {Hypothesis}
        \deductionb
        {B \equiv \timpl{A}{C}}
        {\hastypeg{\hastype{x}{A}}{\pAbs{X}{t}} {C}}
        {2, Lemma \ref{psi-gen}}
        \deductionc
        {C \equiv \tfor{X}{D}}
        {\hastypeg{\hastype{x}{A}}{t} {D}}
        {X \not \in \ftva{\Gamma} \cup \ftva{A}}
        {3, Lemma \ref{psi-gen}}
        \deductiona
        {\tfor{X}{(\timpl{A}{D})} \equiv \timpl{A}{\tfor{X}{D}}}
        {1, Iso. \eqref{iso:pcommi}}
        \deductiona
        {\timpl{A}{C} \equiv \timpl{A}{\tfor{X}{D}}}
        {4, congr. \req}
      \item \hfill
        \begin{center}
          \begin{prooftree}
            \hypo{\judg{\hastype{x}{A}}{t}{D}}
            \infer1[\impli]{\judi{\abs{x^A}{t}} {\timpl{A}{D}}}
            \infer[left label=[4]]1[\fori]{\judi{\pAbs{X}{\abs{x^A}{t}}} {\tfor{X}{(\timpl{A}{D})}}}
            \infer[left label=[5]]1[\req]{\judi{\pAbs{X}{\abs{x^A}{t}}} {\timpl{A}{\tfor{X}{D}}}}
            \infer[left label=[6]]1[\req]{\judi{\pAbs{X}{\abs{x^A}{t}}} {\timpl{A}{C}}}
            \infer[left label=[3]]1[\req]{\judi{\pAbs{X}{\abs{x^A}{t}}} {B}}
          \end{prooftree}
        \end{center} 
        \deductionse
        
      \end{description}
      
    \item[\eqref{spcommei}:]
      $\pApp{(\abs{x^A}{t})}{B} \rightleftarrows \abs{x^A}{\pApp{t}{B}}$
      \begin{description}
      \item[$(^\rightarrow)$]~        
        \deductionsb
        \deductiona
        {X \not \in \ftva{A}}
        {Hypothesis}
        \deductiona
        {\hastypei{\pApp{(\abs{x^A}{t})}{B}} {C}}
        {Hypothesis}
        \deductionb
        {C \equiv \substa{X}{B}{D}}
        {\hastypei{\abs{x^A}{t}} {\tfor{X}{D}}}
        {2, Lemma \ref{psi-gen}}
        \deductionb
        {\tfor{X}{D} \equiv \timpl{A}{E}}
        {\hastypeg{\hastype{x}{A}}{t}{E}}
        {3, Lemma \ref{psi-gen}}
        \deductionb
        {E \equiv \tfor{X}{E'}}
        {D \equiv \timpl{A}{E'}}
        {4, Lemma \ref{psi-equiv-tfor-timpl}}
        \deductiona
        {\timpl{A}{\substa{X}{B}{E'}} = \subst{X}{B}{\timpl{A}{E'}}}
        {1, Def.}
        \deductiona
        {\subst{X}{B}{\timpl{A}{E'}} \equiv \substa{X}{B}{D}}
        {5, congr. \req}
      \item \hfill
        \begin{center}
          \begin{prooftree}
            \hypo{\judg{\hastype{x}{A}}{t}{E}}
            \infer[left label=[5]]1[\req]{\judg{\hastype{x}{A}}{t} {\tfor{X}{E'}}}
            \infer1[\fore]{\judg{\hastype{x}{A}}{\pApp{t}{B}} {\substa{X}{B}{E'}}}
            \infer1[\impli]{\judi{\abs{x^A}{\pApp{t}{B}}} {\timpl{A}{\substa{X}{B}{E'}}}}
            \infer[left label=[6]]1[\req]{\judi{\abs{x^A}{\pApp{t}{B}}} {\subst{X}{B}{\timpl{A}{E'}}}}
            \infer[left label=[7]]1[\req]{\judi{\abs{x^A}{\pApp{t}{B}}} {\substa{X}{B}{D}}}
            \infer[left label=[3]]1[\req]{\judi{\abs{x^A}{\pApp{t}{B}}} {C}}
            
          \end{prooftree}
        \end{center} 
        \deductionse
        
      \item[$(_\leftarrow)$]~ 
        \deductionsb
        \deductiona
        {X \not \in \ftva{A}}
        {Hypothesis}
        \deductiona
        {\hastypei{\abs{x^A}{\pApp{t}{B}}} {C}}
        {Hypothesis}
        \deductionb
        {C \equiv \timpl{A}{D}}
        {\hastypeg{\hastype{x}{A}}{\pApp{t}{B}} {D}}
        {1, Lemma \ref{psi-gen}}
        \deductionb
        {D \equiv \substa{X}{B}{E}}
        {\hastypeg{\hastype{x}{A}}{t}{\tfor{X}{E}}}
        {2, Lemma \ref{psi-gen}}
        \deductiona
        {\timpl{A}{\tfor{X}{E}} \equiv \tfor{X}{(\timpl{A}{E})}}
        {Iso. \eqref{iso:pcommi}}
        \deductiona
        {\subst{X}{B}{\timpl{A}{E}} = \timpl{A}{\substa{X}{B}{E}}}
        {1, Def.}
        \deductiona
        {\timpl{A}{\substa{X}{B}{E}} \equiv \timpl{A}{D}}
        {4, congr. \req}
      \item \hfill
        \begin{center}
          \begin{prooftree}
            \hypo{\judg{\hastype{x}{A}}{t}{\tfor{X}{E}}}
            \infer1[\impli]{\judi{\abs{x^A}{t}}{\timpl{A}{\tfor{X}{E}}}}
            \infer[left label=[5]]1[\req]{\judi{\abs{x^A}{t}}{\tfor{X}{(\timpl{A}{E})}}}
            \infer1[\fore]{\judi{\pApp{(\abs{x^A}{t})}{B}}{\subst{X}{B}{\timpl{A}{E}}}}
            \infer[left label=[6]]1[\req]{\judi{\pApp{(\abs{x^A}{t})}{B}}{\timpl{A}{\substa{X}{B}{E}}}}
            \infer[left label=[7]]1[\req]{\judi{\pApp{(\abs{x^A}{t})}{B}}{\timpl{A}{D}}}
            \infer[left label=[3]]1[\req]{\judi{\pApp{(\abs{x^A}{t})}{B}}{C}}
          \end{prooftree}
        \end{center}
        \deductionse
        
      \end{description}
      
    \item[\eqref{spdistii}:]
      ${\pAbs{X}{\pair{t}{r}} \rightleftarrows \pair{\pAbs{X}{t}}{\pAbs{X}{r}}}$
      \begin{description}
      \item[$(^\rightarrow)$]~        
        \deductionsb
        \deductiona
        {\hastypei{\pAbs{X}{\pair{t}{r}}} {A}}
        {Hypothesis}
        \deductionc
        {A \equiv \tfor{X}{B}}
        {\hastypei{\pair{t}{r}} {B}}
        {X \not \in \ftva{\Gamma}}
        {1, Lemma \ref{psi-gen}}
        \deductionc
        {B \equiv \tconj{C}{D}}
        {\hastypei{t}{C}}
        {\hastypei{r}{D}}
        {2, Lemma \ref{psi-gen}}
        \deductiona
        {\tfor{X}{(\tconj{C}{D})} \equiv \tconj{\tfor{X}{C}}{\tfor{X}{D}}}
        {Iso. \eqref{iso:pdist}}
        \deductiona
        {\tfor{X}{B} \equiv \tfor{X}{(\tconj{C}{D})}}
        {3, congr. \req}
      \item~ 

	\scalebox{0.88}{\centering
          \begin{prooftree}
            \hypo{\judi{t}{C}}
            \infer[left label=[2]]1[\fori]{\judi{\pAbs{X}{t}} {\tfor{X}{C}}}
            \hypo{\judi{r}{D}}
            \infer[left label=[2]]1[\fori]{\judi{\pAbs{X}{r}} {\tfor{X}{D}}}
            \infer2[\conji]{\judi{\pair{\pAbs{X}{t}}{\pAbs{X}{r}}} {\tconj{\tfor{X}{C}}{\tfor{X}{D}}}}
            \infer[left label=[4]]1[\req]{\judi{\pair{\pAbs{X}{t}}{\pAbs{X}{r}}} {\tfor{X}({\tconj{C}{D})}}}
            \infer[left label=[5]]1[\req]{\judi{\pair{\pAbs{X}{t}}{\pAbs{X}{r}}} {\tfor{X}{B}}}
            \infer[left label=[2]]1[\req]{\judi{\pair{\pAbs{X}{t}}{\pAbs{X}{r}}} {A}}
          \end{prooftree}
	}
        \deductionse
        
      \item[$(_\leftarrow)$]~ 
        \deductionsb
        \deductiona
        {\hastypei{\pair{\pAbs{X}{t}}{\pAbs{X}{r}}} {A}}
        {Hypothesis}
        \deductionc
        {A \equiv \tconj{B}{C}}
        {\hastypei{\pAbs{X}{t}} {B}}
        {\hastypei{\pAbs{X}{r}} {C}}
        {1, Lemma \ref{psi-gen}}
        \deductionc
        {B \equiv \tfor{X}{D}}
        {\hastypei{t}{D}}
        {X \not \in \ftva{\Gamma}}
        {2, Lemma \ref{psi-gen}}
        \deductionc
        {C \equiv \tfor{X}{E}}
        {\hastypei{r}{E}}
        {X \not \in \ftva{\Gamma}}
        {2, Lemma \ref{psi-gen}}
        \deductiona
        {\tfor{X}{(\tconj{D}{E})} \equiv \tconj{\tfor{X}{D}}{\tfor{X}{E}}}
        {Iso. \eqref{iso:pdist}}
        \deductiona
        {\tconj{\tfor{X}{D}}{\tfor{X}{E}} \equiv \tconj{B}{C}}
        {3, 4, congr. \req}
      \item \hfill
        \begin{center}
          \begin{prooftree}
            \hypo{\judi{t}{D}}
            \hypo{\judi{r}{E}}
            \infer2[\conji]{\judi{\pair{t}{r}} {\tconj{D}{E}}}
            \infer[left label=[3]]1[\fori]{\judi{\pAbs{X}{\pair{t}{r}}} {\tfor{X}{(\tconj{D}{E})}}}
            \infer[left label=[5]]1[\req]{\judi{\pAbs{X}{\pair{t}{r}}} {\tconj{\tfor{X}{D}}{\tfor{X}{E}}}}
            \infer[left label=[6]]1[\req]{\judi{\pAbs{X}{\pair{t}{r}}} {\tconj{B}{C}}}
            \infer[left label=[2]]1[\req]{\judi{\pAbs{X}{\pair{t}{r}}} {A}}
          \end{prooftree}
        \end{center} 
        \deductionse
      \end{description}
      
    \item[\eqref{spdistei}:]
      $\pApp{\pair{t}{r}}{B} \rightleftarrows \pair{\pApp{t}{B}}{\pApp{r}{B}}$
      \begin{description}
      \item[$(^\rightarrow)$]~
        \deductionsb
        \deductiona
        {\hastypei{\pApp{\pair{t}{r}}{B}} {A}}
        {Hypothesis}
        \deductionb
        {A \equiv \substa{X}{B}{C}}
        {\hastypei{\pair{t}{r}} {\tfor{X}{C}}}
        {1, Lemma \ref{psi-gen}}
        \deductionc
        {\tfor{X}{C} \equiv \tconj{D}{E}}
        {\hastypei{t}{D}}
        {\hastypei{r}{E}}
        {2, Lemma \ref{psi-gen}}
        \deductionc
        {D \equiv \tfor{X}{D'}}
        {E \equiv \tfor{X}{E'}}
        {C \equiv \tconj{D'}{E'}}
        {3, Lemma \ref{psi-equiv-tfor-tconj}}
        \deductiona
        {\subst{X}{B}{\tconj{D'}{E'}} = \tconj{\substa{X}{B}{D'}}{\substa{X}{B}{E'}}}
        {Def.}
        \deductiona
        {\substa{X}{B}{C} \equiv \subst{X}{B}{\tconj{D'}{E'}}}
        {4, congr. \req}
      \item~ 

	\scalebox{0.84}{\centering
          \begin{prooftree}
            \hypo{\judi{t}{D}}
            \infer[left label=[4]]1[\req]{\judi{t} {\tfor{X}{D'}}}
            \infer1[\fore]{\judi{\pApp{t}{B}} {\substa{X}{B}{D'}}}
            \hypo{\judi{r}{E}}
            \infer[left label=[4]]1[\req]{\judi{r} {\tfor{X}{E'}}}
            \infer1[\fore]{\judi{\pApp{r}{B}} {\substa{X}{B}{E'}}}
            \infer2[\conji]{\judi{\pair{\pApp{t}{B}}{\pApp{r}{B}}} {\tconj{\substa{X}{B}{D'}}{\substa{X}{B}{E'}}}}
            \infer[left label=[5]]1[\req]{\judi{\pair{\pApp{t}{B}}{\pApp{r}{B}}} {\subst{X}{B}{\tconj{D'}{E'}}}}
            \infer[left label=[6]]1[\req]{\judi{\pair{\pApp{t}{B}}{\pApp{r}{B}}} {\substa{X}{B}{C}}}
            \infer[left label=[2]]1[\req]{\judi{\pair{\pApp{t}{B}}{\pApp{r}{B}}} {A}}
          \end{prooftree}
        } 
        \deductionse
        
      \item[$(_\leftarrow)$]~
        \deductionsb
        \deductiona
        {\hastypei{\pair{\pApp{t}{B}}{\pApp{r}{B}}} {A}}
        {Hypothesis}
        \deductionc
        {A \equiv \tconj{C}{D}}
        {\hastypei{\pApp{t}{B}} {C}}
        {\hastypei{\pApp{r}{B}} {D}}
        {1, Lemma \ref{psi-gen}}
        \deductionb
        {C \equiv \substa{X}{B}{C'}}
        {\hastypei{t}{\tfor{X}{C'}}}
        {2, Lemma \ref{psi-gen}}
        \deductionb
        {D \equiv \substa{X}{B}{D'}}
        {\hastypei{r}{\tfor{X}{D'}}}
        {2, Lemma \ref{psi-gen}}
        \deductiona
        {\tfor{X}{(\tconj{C'}{D'})} \equiv \tconj{\tfor{X}{C'}}{\tfor{X}{D'}}}
        {Iso. \eqref{iso:pdist}}
        \deductiona
        {\subst{X}{B}{\tconj{C'}{D'}} = \tconj{\substa{X}{B}{C'}}{\substa{X}{B}{D'}}}
        {Def.}
        \deductiona
        {\tconj{\substa{X}{B}{C'}}{\substa{X}{B}{D'}} \equiv \tconj{C}{D}}
        {3, 4, congr. \req}
	\needspace{3em}
      \item \hfill
        \begin{center}
          \begin{prooftree}
            \hypo{\judi{t} {\tfor{X}{C'}}}
            \hypo{\judi{r} {\tfor{Y}{D'}}}
            \infer2[\conji]{\judi{\pair{t}{r}} {\tconj{\tfor{X}{C'}}{\tfor{X}{D'}}}}
            \infer[left label=[5]]1[\req]{\judi{\pair{t}{r}} {\tfor{X}{(\tconj{C'}{D'}})}}
            \infer1[\fore]{\judi{\pApp{\pair{t}{r}}{B}} {\subst{X}{B}{\tconj{C'}{D'}}}}
            \infer[left label=[6]]1[\req]{\judi{\pApp{\pair{t}{r}}{B}} {\tconj{\substa{X}{B}{C'}}{\substa{X}{B}{D'}}}}
            \infer[left label=[7]]1[\req]{\judi{\pApp{\pair{t}{r}}{B}} {\tconj{{C}}{{D}}}}
            \infer[left label=[2]]1[\req]{\judi{\pApp{\pair{t}{r}}{B}} {A}}
          \end{prooftree}
        \end{center}
        \deductionse
        
      \end{description}
      
    \item[\eqref{spdistie}:]
      $\proj{\tfor{X}{B}}{\pAbs{X}{t}} \rightleftarrows \pAbs{X}{\proja{B}{t}}$
      \begin{description}
      \item[$(^\rightarrow)$]~
        \deductionsb
        \deductiona
        {\hastypei{\proj{\tfor{X}{B}}{\pAbs{X}{t}}} {A}}
        {Hypothesis}
        \deductionb
        {A \equiv \tfor{X}{B}}
        {\hastypei{\pAbs{X}{t}} {\tconj{(\tfor{X}{B})}{C}}}
        {1, Lemma \ref{psi-gen}}
        \deductionc
        {\tconj{(\tfor{X}{B})}{C} \equiv \tfor{X}{D}}
        {\hastypei{t}{D}}
        {X \not \in \ftva{\Gamma}}
        {2, Lemma \ref{psi-gen}}
        \deductionb
        {C \equiv \tfor{X}{C'}}
        {D \equiv \tconj{B}{C'}}
        {3, Lemma \ref{psi-equiv-tfor-tconj}}
      \item \hfill
        \begin{center}
          \begin{prooftree}
            \hypo{\judi{t}{D}}
            \infer[left label=[4]]1[\req]{\judi{t} {\tconj{B}{C'}}}
            \infer1[\conje]{\judi{\proja{B}{t}} {B}}
            \infer[left label=[3]]1[\fori]{\judi{\pAbs{X}{\proja{B}{t}}} {\tfor{X}{B}}}
            \infer[left label=[2]]1[\req]{\judi{\pAbs{X}{\proja{B}{t}}} {A}}
          \end{prooftree}
        \end{center}
        \deductionse
        
      \item[$(_\leftarrow)$]~
        \deductionsb
        \deductiona
        {\hastypei{\pAbs{X}{\proja{B}{t}}} {A}}
        {Hypothesis}
        \deductionc
        {A \equiv \tfor{X}{C}}
        {\hastypei{\proja{B}{t}} {C}}
        {X \not \in \ftva{\Gamma}}
        {1, Lemma \ref{psi-gen}}
        \deductionb
        {B \equiv C}
        {\hastypei{t}{\tconj{C}{D}}}
        {2, Lemma \ref{psi-gen}}
        \deductiona
        {\tfor{X}{(\tconj{C}{D})} \equiv \tconj{\tfor{X}{C}}{\tfor{X}{D}}}
        {Iso. \eqref{iso:pdist}}
      \item \hfill
        \begin{center}
          \begin{prooftree}
            \hypo{\judi{t} {\tconj{C}{D}}}
            \infer[left label=[2]]1[\fori]{\judi{\pAbs{X}{t}} {\tfor{X}{(\tconj{C}{D})}}}
            \infer[left label=[4]]1[\req]{\judi{\pAbs{X}{t}} {\tconj{\tfor{X}{C}}{\tfor{X}{D}}}}
            \infer1[\conje]{\judi{\proj{\tfor{X}{B}}{\pAbs{X}{t}}} {\tfor{X}{C}}}
            \infer[left label=[2]]1[\req]{\judi{\proj{\tfor{X}{B}}{\pAbs{X}{t}}} {A}}
          \end{prooftree}
        \end{center}
        \deductionse
        
      \end{description}
      
    \item[\eqref{spdistee}:]
      $\pApp{(\proja{\tfor{X}{B}}{t})}{C} \rightleftarrows \proj{\substa{X}{C}{B}}{\pApp{t}{C}}$
      \begin{description}
      \item[$(^\rightarrow)$]~        
        \deductionsb
        \deductiona
        {\hastypei{t}{\tfor{X}{(\tconj{B}{D})}}}
        {Hypothesis}
        \deductiona
        {\hastypei{\pApp{(\proja{\tfor{X}{B}}{t})}{C}} {A}}
        {Hypothesis}
        \deductionb
        {A \equiv \substa{X}{C}{E}}
        {\hastypei{\proja{\tfor{X}{B}}{t}} {\tfor{X}{E}}}
        {2, Lemma \ref{psi-gen}}
        \deductionb
        {\tfor{X}{E} \equiv \tfor{X}{B}}
        {\hastypei{t} {\tconj{\tfor{X}{E}}{F}}}
        {3, Lemma \ref{psi-gen}}
        \deductiona
        {E \equiv B}
        {4}
        \deductiona
        {\subst{X}{C}{\tconj{B}{D}} = \tconj{\substa{X}{C}{B}}{\substa{X}{C}{D}}}
        {Def.}
        \deductiona
        {\substa{X}{C}{B} \equiv \substa{X}{C}{E}}
        {5, congr. \req}
	\needspace{3em}
      \item \hfill
        \begin{center}
          \begin{prooftree}
            \hypo{\judi{t}{\tfor{X}{\tconj{B}{D}}}}
            \infer1[\fore]{\judi{\pApp{t}{C}} {\subst{X}{C}{\tconj{B}{D}}}}
            \infer[left label=[6]]1[\req]{\judi{\pApp{t}{C}} {\tconj{\substa{X}{C}{B}}{\substa{X}{C}{D}}}}
            \infer1[\conje]{\judi{\proj{\substa{X}{C}{B}}{\pApp{t}{C}}} {\substa{X}{C}{B}}}
            \infer[left label=[7]]1[\req]{\judi{\proj{\substa{X}{C}{B}}{\pApp{t}{C}}} {\substa{X}{C}{E}}}
            \infer[left label=[3]]1[\req]{\judi{\proj{\substa{X}{C}{B}}{\pApp{t}{C}}} {A}}
          \end{prooftree}
        \end{center}
        \deductionse
        
      \item[$(_\leftarrow)$]~
        \deductionsb
        \deductiona
        {\hastypei{t}{\tfor{X}{(\tconj{B}{D})}}}
        {Hypothesis}
        \deductiona
        {\hastypei{\proj{\substa{X}{C}{B}}{\pApp{t}{C}}} {A}}
        {Hypothesis}
        \deductionb
        {A \equiv \substa{X}{C}{B}}
        {\hastypei{\pApp{t}{C}} {\tconj{A}{E}}}
        {2, Lemma \ref{psi-gen}}
        \deductiona
        {\tfor{X}{(\tconj{B}{D})} \equiv \tconj{\tfor{X}{B}}{\tfor{X}{D}}}
        {Iso. \eqref{iso:pdist}}
      \item \hfill
        \begin{center}
          \begin{prooftree}
            \hypo{\judi{t}{\tfor{X}{(\tconj{B}{D})}}}
            \infer[left label=[4]]1[\req]{\judi{t} {\tconj{\tfor{X}{B}}{\tfor{X}{D}}}}
            \infer1[\conje]{\judi{\proja{\tfor{X}{B}}{t}} {\tfor{X}{B}}}
            \infer1[\fore]{\judi{\pApp{(\proja{\tfor{X}{B}}{t})}{C}} {\substa{X}{C}{B}}}
            \infer[left label=[3]]1[\req]{\judi{\pApp{(\proja{\tfor{X}{B}}{t})}{C}} {A}}
          \end{prooftree}
        \end{center}
        \deductionse
        
      \end{description}
      
    \item[\eqref{red:smallbeta}:] If \hastypei{s}{A}, \app{(\abs{x^A}{r})}{s}  $\hookrightarrow$ \substa{x}{s}{r}
      \deductionsb
      \deductiona{\hastypei{s}{A}}{Hypothesis}
      \deductiona{\hastypei{\abs{x^A}{r}}{B}}{Hypothesis}
      \deductiona{\hastypei{\abs{x^A}{r}}{\timpl{A}{B}}}{2, Lemma \ref{psi-gen}}
      \deductionb{\timpl{A}{B} \equiv \timpl{A}{C}}{\hastypeg{\hastype{x}{A}}{r}{C}}{3, Lemma \ref{psi-gen}}
      \deductiona{B \equiv C}{4, congr. \req}
      \deductiona{\hastypei{\substa{x}{s}{r}}{C}}{1, 4, Lemma \ref{substitution}}
      \deductiona{\hastypei{\substa{x}{s}{r}}{B}}{5, 6, rule \req}
      \deductionse
      \strut \hfill
      
    \item[\eqref{red:bigbeta}:] $(\Lambda X.r)[A] \hookrightarrow [X:=A]r$
      \deductionsb
      \deductiona{\hastypei{\pApp{(\pAbs{X}{r})}{A}}{B}}{Hypothesis}
      \deductionb{B \equiv \substa{X}{A}{C}}{\hastypei{\pAbs{X}{r}}{\tfor{X}{C}}}{1, Lemma \ref{psi-gen}}
      \deductionc{\tfor{X}{C} \equiv \tfor{X}{D}}{\hastypei{r}{D}}{X \not \in \ftva{\Gamma}}{2, Lemma \ref{psi-gen}}
      \deductiona{C \equiv D}{3}
      \deductiona{\hastypei{r}{C}}{4, rule \req}
      \deductiona{\substa{X}{A}{\Gamma} \vdash \hastypei{\substa{X}{A}{r}}{\substa{X}{A}{C}}}{5, Lemma \ref{substitution}}
      \deductiona{\hastypei{\substa{X}{A}{r}}{B}}{2, 3, 7, rule \req}
      \deductionse 
      \strut \hfill
      
    \item[\eqref{red:pi}:] If \hastypei{r}{A}, \proja{A}{\pair{r}{s}} $\hookrightarrow$ $r$
      \deductionsb
      \deductiona{\hastypei{r}{A}}{Hypothesis}
      \deductiona{\hastype{\proja{A}{\pair{r}{s}}}{B}}{Hypothesis}
      \deductionb{B \equiv A}{\hastypei{\p{r}{s}}{\tconj{A}{C}}}{2, Lemma \ref{psi-gen}}
      \deductiona{\hastypei{\proja{A}{\pair{r}{s}}}{A}}{2, 3, rule \req}
      \deductionse
      \qedhere
    \end{description}
\end{proof}

\section{Detailed proof of Section~\ref{sec:RP}}\label{app:RP}
    
\recap{Lemma}{lem:eqProd}
{
  For all $r, s, t$ such that $\pair{r}{s}\eq^*t$, we have either
  \begin{enumerate}
  \item $t=\pair{u}{v}$ where either
    \begin{enumerate}
    \item $u\eq^*\pair{t_{11}}{t_{21}}$ and $v\eq^*       \pair{t_{12}}{t_{22}}$ with $r\eq^*\pair{t_{11}}{t_{12}}$ and
      $s\eq^*\pair{t_{21}}{t_{22}}$, or
    \item $v\eq^*\pair{w}{s}$ with $r\eq^*\pair{u}{w}$, or
      any of the three symmetric cases, or
    \item $r\eq^*u$ and $s\eq^*v$, or the symmetric case.
    \end{enumerate}
  \item $t=\lambda x^A.a$ and $a\eq^*\pair{a_1}{a_2}$ with $r\eq^*\lambda x^A.a_1$ and $s\eq^*\lambda x^A.a_2$.
  \item $t=av$ and $a\eq^*\pair{a_1}{a_2}$,
    with $r\eq^*a_1v$ and $s\eq^*a_2v$.
  \item $t=\Lambda X.a$ and $a\eq^*\pair{a_1}{a_2}$ with $r\eq^*\Lambda X.a_1$ and $s\eq^*\Lambda X.a_2$.
  \item $t=\tapp{a}{A}$ and $a\eq^*\pair{a_1}{a_2}$, with $r\eq^*a_1[A]$ and $s\eq^*a_2[A]$.
  \end{enumerate}
}
\begin{proof}
  By a double induction, first on $M( t)$ and then on the length of the relation
  $\eq^*$. Consider an equivalence proof $ \pair{r}{s} \eq^* t' \eq   t$ with a shorter proof $ \pair{r}{s} \eq^* t'$. By the second
  induction hypothesis, the term $ t'$ has the form prescribed by the lemma.
  We consider the five cases and in each case, the possible rules transforming
  $ t'$ in $ t$.
  \begin{enumerate}
  \item Let $ \pair{r}{s}\eq^* \pair{u}{v}\eq t$. The possible
    equivalences from $ \pair{u}{v}$ are
    \begin{itemize}
    \item $ t= \pair{u'}{v}$ or $ \pair{u}{v'}$ with $ u\eq       u'$ and $ v\eq v'$, and so the term $ t$ is in case \ref{it:eqProd-sum}.
    \item Rules \rulelabel{(comm)} and \rulelabel{(asso)} preserve the conditions of case
      \ref{it:eqProd-sum}.
    \item $ t=\lambda x^A.\pair{u'}{v'}$, with $ u=\lambda x^A.       u'$ and $ v=\lambda x^A. v'$, and so the term $t$ is in case \ref{it:eqProd-lam}.
    \item $ t=\pair{u'}{v'} t'$, with $ u= u' t'$ and
      $ v= v' t'$, and so the term $t$ is in case \ref{it:eqProd-app}.
    \item $ t=\Lambda X.\pair{u'}{v'}$, with $ u=\Lambda X.
      u'$ and $ v=\Lambda X. v'$, and so the term $t$ is in case \ref{it:eqProd-tlam}.
    \item $ t=\pair{u'}{v'} [A]$, with $ u= u' [A]$ and
      $ v= v' [A]$, and so the term $t$ is in case \ref{it:eqProd-tapp}.
    \end{itemize}

  \item Let $ \pair{r}{s}\eq^*\lambda x^A. a\eq t$, with $     a\eq^*\pair{a_1}{a_2}$, $ r\eq^*\lambda x^A. a_1$, and $     s\eq^*\lambda x^A. a_2$. Hence, possible equivalences from $\lambda
    x^A. a$ to $ t$ are
    \begin{itemize}
    \item $ t=\lambda x^A. a'$ with $ a\eq^* a'$, hence $       a'\eq \pair{a_1}{a_2}$, and so the term $t$ is in case \ref{it:eqProd-lam}.
    \item $ t=\pair{\lambda x^A. u}{\lambda x^A. v}$, with $ \pair{a_1}{a_2}\eq^* a= \pair{u}{v}$. Hence, by the first induction
      hypothesis (since $M( a)<M( t)$), either
      \begin{enumerate}
      \item $ a_1\eq^* u$ and $ a_2\eq^* v$, and so $         r\eq^*\lambda x^A. u$ and $ s\eq^*\lambda x^A. v$, or
      \item $ v\eq^*  \pair{t_1}{t_2}$ with $ a_1\eq^* \pair{u}{t_1}$ and $ a_2\eq^* t_2$, and so $\lambda x^A. v\eq^*\pair{\lambda x^A. t_1}{\lambda x^A. t_2}$, $ r\eq^*\pair{\lambda x^A.         u}{\lambda x^A. t_1}$ and $ s\eq^*\lambda x^A. t_2$, or
      \item $ u\eq^* \pair{t_{11}}{t_{21}}$ and $ v\eq^*         \pair{t_{12}}{t_{22}}$ with $ a_1\eq^* \pair{t_{11}}{t_{12}}$
        and $ a_2\eq^* \pair{t_{21}}{t_{22}}$, and so $\lambda x^A.         u\eq^*\pair{\lambda x^A. t_{11}}{\lambda x^A. t_{21}}$, $\lambda
        x^A. v\eq^*\pair{\lambda x^A. t_{12}}{\lambda x^A. t_{22}}$, $         r\eq^*\pair{\lambda x^A. t_{11}}{\lambda x^A. t_{12}}$ and $         s\eq^*\pair{\lambda x^A. t_{21}}{\lambda x^A. t_{22}}$.
      \end{enumerate} (the symmetric cases are analogous), and so the term $t$
      is in case \ref{it:eqProd-sum}.
    \item $ t=\Lambda X. \lambda x^A. a'$ with $a = \Lambda X. a'$, hence $\Lambda X. a'\eq^* \pair{a_1}{a_2}$. Since $M(\pair{a_1}{a_2}) < M(\pair{r}{s})$, by the first induction hypothesis, the term $t$ is in case \ref{it:eqProd-tlam}.
    \item $t = \abs{x^A}{a'}[B]$ with $a = a'[B]$, hence $a'[B] \eq^* \pair{a_1}{a_2}$. Since $M(\pair{a_1}{a_2}) < M(\pair{r}{s})$, by the first induction hypothesis, the term $t$ is in case \ref{it:eqProd-tapp}.
    \end{itemize}
  \item Let $ \pair{r}{s}\eq^* a w\eq t$, with $ a\eq^*     \pair{a_1}{a_2}$, $ r\eq^* a_1 w$, and $ s\eq^* a_2 w$.
    The possible equivalences from $ a w$ to $ t$ are
    \begin{itemize}
    \item $ t= a' w$ with $ a\eq^* a'$, hence $ a'\eq^*       \pair{a_1}{a_2}$, and so the term $t$ is in case \ref{it:eqProd-app}.
    \item $ t= a w'$ with $ w\eq^* w'$ and so the term $t$ is in case \ref{it:eqProd-app}.
    \item $ t= \pair{u w}{v w}$, with $ \pair{a_1}{a_2}
      a_2\eq^* a= \pair{u}{v}$. Hence, by the first induction hypothesis
      (since $M( a)<M( t)$), either
      \begin{enumerate}
      \item $ a_1\eq^* u$ and $ a_2\eq^* v$, and so $ r\eq^*         u w$ and $ s\eq^* v w$, or
      \item $ v\eq^*  \pair{t_1}{t_2}$ with $ a_1\eq^* \pair{u}{t_1}$ and $ a_2\eq^* t_2$, and so $ v w\eq^* \pair{t_1 w}{t_2 w}$, $ r\eq^* \pair{u w}{t_1 w}$ and
        $ s\eq^* t_2 w$, or
      \item $ u\eq^* \pair{t_{11}}{t_{21}}$ and $ v\eq^*         \pair{t_{12}}{t_{22}}$ with $ a_1\eq^* \pair{t_{11}}{t_{12}}$
        and $ a_2\eq^* \pair{t_{21}}{t_{22}}$, and so $ u         w\eq^* \pair{t_{11} w}{t_{21} w}$, $ v w\eq^*         \pair{t_{12} w}{t_{22} w}$, $ r\eq^* \pair{t_{11} w}{t_{12} w}$ and $ s\eq^* \pair{t_{21} w}{t_{22} w}$.
      \end{enumerate} (the symmetric cases are analogous), and so the term $t$
      is in case \ref{it:eqProd-sum}.
    \item $ t= a'\pair{v}{w}$ with $ a= a' v$, thus $       a' v= a\eq^* \pair{a_1}{a_2}$. Hence, by the first induction
      hypothesis, $ a'\eq^* \pair{a'_1}{a'_2}$, with $ a_1\eq^*       a'_1 v$ and $ a_2\eq^* a'_2 v$. Therefore, $ r\eq^*       a'_1\pair{v}{w}$ and $ s\eq^* a'_2\pair{v}{w}$, and
      so the term $t$ is in case \ref{it:eqProd-app}.
    \end{itemize}
  \item Let $\pair{r}{s}\eq^* \Lambda X. a\eq t$, with $     a\eq^*\pair{a_1}{a_2}$, $ r\eq^*\Lambda X. a_1$, and $     s\eq^*\Lambda X. a_2$. Hence, possible equivalences from $\lambda
    x. a$ to $ t$ are
    \begin{itemize}
    \item $ t=\Lambda X. a'$ with $ a\eq^* a'$, hence $       a'\eq \pair{a_1}{a_2}$, and so the term $t$ is in case \ref{it:eqProd-tlam}.
    \item $ t=\pair{\Lambda X. u}{\Lambda X. v}$, with $ \pair{a_1}{a_2}\eq^* a= \pair{u}{v}$. Hence, by the first induction
      hypothesis (since $M( a)<M( t)$), either
      \begin{enumerate}
      \item $ a_1\eq^* u$ and $ a_2\eq^* v$, and so $         r\eq^*\Lambda X. u$ and $ s\eq^*\Lambda X. v$, or
      \item $ v\eq^*  \pair{t_1}{t_2}$ with $ a_1\eq^* \pair{u}{t_1}$ and $ a_2\eq^* t_2$, and so $\Lambda X. v\eq^*\pair{\Lambda        X. t_1}{\Lambda X. t_2}$, $ r\eq^*\pair{\Lambda X.         u}{\Lambda X. t_1}$ and $ s\eq^*\Lambda X. t_2$, or
      \item $ u\eq^* \pair{t_{11}}{t_{21}}$ and $ v\eq^*         \pair{t_{12}}{t_{22}}$ with $ a_1\eq^* \pair{t_{11}}{t_{12}}$
        and $ a_2\eq^* \pair{t_{21}}{t_{22}}$, and so $\Lambda X.         u\eq^*\pair{\Lambda X. t_{11}}{\Lambda X. t_{21}}$, $\Lambda        X. v\eq^*\pair{\Lambda X. t_{12}}{\Lambda X. t_{22}}$, $         r\eq^*\pair{\Lambda X. t_{11}}{\Lambda X. t_{12}}$ and $         s\eq^*\pair{\Lambda X. t_{21}}{\Lambda X. t_{22}}$.
      \end{enumerate} (the symmetric cases are analogous), and so the term $t$
      is in case \ref{it:eqProd-sum}.
      \item $ t=\lambda x^A. \Lambda X. a'$ with $a = \lambda x^A. a'$, hence $\lambda x^A. a'\eq^* \pair{a_1}{a_2}$. Since $M(\pair{a_1}{a_2}) < M(\pair{r}{s})$, by the first induction hypothesis, the term $t$ is in case \ref{it:eqProd-lam}.
      \item $t = (\Lambda X.a')[B]$ with $a = a'[B]$, hence $a'[B] \eq^* \pair{a_1}{a_2}$. Since $M(\pair{a_1}{a_2}) < M(\pair{r}{s})$, by the first induction hypothesis, the term $t$ is in case \ref{it:eqProd-tapp}.
    \end{itemize}

    \item Let $ \pair{r}{s}\eq^* a [A]\eq t$, with $ a\eq^* \pair{a_1}{a_2}$, $ r\eq^* a_1 [A]$, and $ s\eq^* a_2 [A]$.
    The possible equivalences from $ a [A]$ to $ t$ are
    \begin{itemize}
      \item $ t= a' [A]$ with $ a\eq^* a'$, hence $ a'\eq^* \pair{a_1}{a_2}$, and so the term $t$ is in case \ref{it:eqProd-tapp}.
      \item $ t= \pair{u [A]}{v [A]}$, with $ \pair{a_1}{a_2}\eq^* a= \pair{u}{v}$. Hence, by the first induction hypothesis
      (since $M( a)<M( t)$), either
      \begin{enumerate}
      \item $ a_1\eq^* u$ and $ a_2\eq^* v$, and so $ r\eq^* u [A]$ and $ s\eq^* v [A]$, or
      \item $ v\eq^*  \pair{t_1}{t_2}$ with $ a_1\eq^* \pair{u}{t_1}$ and $ a_2\eq^* t_2$, and so $ v [A]\eq^* \pair{t_1 [A]}{t_2 [A]}$, $ r\eq^* \pair{u [A]}{t_1 [A]}$ and
        $ s\eq^* t_2 [A]$, or
      \item $ u\eq^* \pair{t_{11}}{t_{21}}$ and $ v\eq^* \pair{t_{12}}{t_{22}}$ with $ a_1\eq^* \pair{t_{11}}{t_{12}}$
        and $ a_2\eq^* \pair{t_{21}}{t_{22}}$, and so $ u [A]\eq^* \pair{t_{11} [A]}{t_{21} [A]}$, $ v [A]\eq^* \pair{t_{12} [A]}{t_{22} [A]}$, $ r\eq^* \pair{t_{11} [A]}{t_{12} [A]}$ and $ s\eq^* \pair{t_{21} [A]}{t_{22} [A]}$.
      \end{enumerate} (the symmetric cases are analogous), and so the term $t$
      is in case \ref{it:eqProd-sum}.
      \item $t = \lambda x^B.(a'[A])$ with $a = \lambda x^B.a'$, hence $\lambda x^B.a' \eq^* \pair{a_1}{a_2}$. Since $M(\lambda x^B.a') < M(\pair{r}{s})$, by the first induction hypothesis, the term $t$ is in case \ref{it:eqProd-lam}.
      \item $t = \pi_{\substa{X}{A}{B}}(a'[A])$ with $a = \pi_{\forall X.B}(a')$, hence $ \pi_{\forall X.B}(a') \eq^* \pair{a_1}{a_2}$. Since $M(\pi_{\forall X.B}(a')) < M(\pair{r}{s})$, by the first induction hypothesis, this case is absurd. Indeed, $ \pair{a_1}{a_2}$ is never equivalent to $\pi_{\forall X.B}(a')$.
      \qedhere
    \end{itemize}
  \end{enumerate}
\end{proof}

\recap{Lemma}{lem:reProd} {
  For all $r_1, r_2, s, t$ such that $\pair{r_1}{r_2}\eq^* s\re t$, there exists
  $ u_1, u_2$ such that $ t\eq^*  \pair{u_1}{u_2}$ and either
  \begin{enumerate}
  \item $ r_1\toreq   u_1$ and $ r_2\toreq u_2$, 
  \item $ r_1\toreq u_1$ and $   r_2\eq^* u_2$, or
  \item $ r_1\eq^* u_1$ and $ r_2\toreq u_2$.
  \end{enumerate}
}
\begin{proof}
  By induction on $M( \pair{r_1}{r_2})$. By Lemma~\ref{lem:eqProd}, $s$ is
  either a product, an abstraction, an application, a type abstraction or a type
  application with the conditions given in the lemma. The different terms $ s$
  reducible by $\re$ are
  \begin{itemize}
  \item $(\lambda x^A. a) s'$ that reduces by the \rulelabel{($\beta_\lambda$)}
    rule to $ \substa{x}{s'}{a}$.
  \item $(\Lambda X. a)[A]$ that reduces by the \rulelabel{($\beta_\Lambda$)}
    rule to \substa{X}{A}{a}.
  \item $ \pair{s_1}{s_2}$, $\lambda x^A. a$, $ a s'$, $\Lambda X.a$,
    $\tapp{a}{A}$ with a reduction in the subterm $ s_1$, $ s_2$, $ a$, or $ s'$.
  \end{itemize}
  Notice that rule \rulelabel{($\pi$)} cannot apply since $s\nneq^*\pi_C( s')$.
  
  We consider each case:
  \begin{itemize}
  \item $ s=(\lambda x^A. a) s'$ and $ t= \substa{x}{s'}{a}$. Using twice
    Lemma~\ref{lem:eqProd}, we have $ a\eq^* \pair{a_1}{a_2}$, $ r_1\eq^*(\lambda
    x^A. a_1) s'$ and $ r_2\eq^*(\lambda x^A. a_2) s'$. Since $ t\eq^*
    \pair{\substa{x}{s'}{a_1}}{\substa{x}{s'}{a_2}}$, we take $ u_1=
    \substa{x}{s'}{a_1}$ and $ u_2= \substa{x}{s'}{a_2}$.
  \item $ s=(\Lambda X. a) [A]$ and $ t= \substa{X}{A}{a}$. Using twice
    Lemma~\ref{lem:eqProd}, we have $ a\eq^* \pair{a_1}{a_2}$, $ r_1\eq^*(\Lambda X.
    a_1)[A]$ and $ r_2\eq^*(\Lambda X. a_2)[A]$. Since $ t\eq^*
    \pair{\substa{X}{A}{a_1}}{\substa{X}{A}{a_2}}$, we take $ u_1=
    \substa{X}{A}{a_1}$ and $ u_2= \substa{X}{A}{a_2}$.
  \item $ s= \pair{s_1}{s_2}$, $ t= \pair{t_1}{s_2}$ or $ t= \pair{s_1}{t_2}$,
    with $ s_1\re t_1$ and $ s_2\re t_2$. We only consider the first case since the
    other is analogous. One of the following cases happen
    \begin{enumerate}
    \item[(a)] $ r_1\eq^* \pair{w_{11}}{w_{21}}$, $ r_2\eq^*
      \pair{w_{12}}{w_{22}}$, $ s_1= \pair{w_{11}}{w_{12}}$ and $ s_2=
      \pair{w_{21}}{w_{22}}$. Hence, by the induction hypothesis, either $ t_1=
      \pair{w'_{11}}{w_{12}}$, or $ t_1= \pair{w_{11}}{w'_{12}}$, or $ t_1=
      \pair{w'_{11}}{w'_{12}}$, with $ w_{11}\re w_{11}'$ and $ w_{12}\re w_{12}'$. We
      take, in the first case $ u_1= \pair{w_{11}'}{w_{21}}$ and $ u_2=
      \pair{w_{12}}{w_{22}}$, in the second case $ u_1= \pair{w_{11}}{w_{21}}$ and $
      u_2= \pair{w_{12}'}{w_{22}}$, and in the third $ u_1= \pair{w'_{11}}{w_{21}}$
      and $ u_2= \pair{w_{12}'}{w_{22}}$.
    \item[(b)] We consider two cases, since the other two are symmetric.
      \begin{itemize}
      \item $ r_1\eq^* \pair{s_1}{w}$ and $ s_2\eq^* \pair{w}{r_2}$, in which
        case we take $ u_1= \pair{t_1}{w}$ and $ u_2= r_2$.
      \item $ r_2\eq^* \pair{w}{s_2}$ and $ s_1= \pair{r_1}{w}$. Hence, by the
        induction hypothesis, either $ t_1= \pair{r'_1}{w}$, or $ t_1= \pair{r_1}{w'}$
        or $ t_1= \pair{r'_1}{w'}$, with $ r_1\re r_1'$ and $ w\re w'$. We take, in the
        first case $ u_1= r'_1$ and $ u_2= \pair{w}{s_2}$, in the second case $ u_1=
        r_1$ and $ u_2= \pair{w'}{s_2}$, and in the third case $ u_1= r'_1$ and $ u_2=
        \pair{w'}{s_2}$.
      \end{itemize}
    \item[(c)] $ r_1\eq^* s_1$ and $ r_2\eq^* s_2$, in which case we take $ u_1=
      t_1$ and $ u_2= s_2$.

    \end{enumerate}
  \item $ s=\lambda x^A. a$, $ t=\lambda x^A. t'$, and $ a\re t'$, with $ a\eq^*
    \pair{a_1}{a_2}$ and $ s\eq^*\lambda x^A. \pair{a_1}{a_2} x^A. a_2$. Therefore,
    by the induction hypothesis, there exists $ u'_1$, $ u'_2$ such that either ($
    a_1\toreq u'_1$ and $ a_2\toreq u'_2$), or ($ a_1\eq^* u'_1$ and $ a_2\toreq
    u'_2$), or ($ a_1\toreq u'_1$ and $ a_2\eq^* u'_2$). Therefore, we take $
    u_1=\lambda x^A. u_1'$ and $ u_2=\lambda x^A. u_2'$.
  \item $ s= a s'$, $ t= t' s'$, and $ a\re t'$, with $ a\eq^* \pair{a_1}{a_2}$
    and $ s\eq^* \pair{a_1 s'}{a_2 s'}$. Therefore, by the induction hypothesis,
    there exists $ u'_1$, $ u'_2$ such that either ($ a_1\toreq u'_1$ and $
    a_2\toreq u'_2$), or ($ a_1\eq^* u'_1$ and $ a_2\toreq u'_2$), or ($ a_1\toreq
    u'_1$ and $ a_2\eq^* u'_2$). Therefore, we take $ u_1= u_1' s'$ and $ u_2= u_2'
    s'$.
  \item $ s= a s'$, $ t= a t'$, and $ s'\re t'$, with $ a\eq^* \pair{a_1}{a_2}$
    and $ s\eq^* \pair{a_1 s'}{a_2 s'}$. By Lemma~\ref{lem:eqProd} several times,
    one the following cases happen
    \begin{enumerate}
    \item[(a)] $ a_1 s'\eq^* \pair{w_{11} s'}{w_{12} s'}$, $ a_2 s'\eq^*
      \pair{w_{21} s'}{w_{22} s'}$, $ r_1\eq^* \pair{w_{11} s'}{w_{21} s'}$ and $
      r_2\eq^* \pair{w_{12} s'}{w_{22} s'}$. We take $ u_1\eq^*(
      \pair{w_{11}}{w_{21}}) t'$ and $ r_2\eq^*( \pair{w_{12}}{w_{22}}) t'$.
    \item[(b)] $ a_2 s'\eq^* \pair{w_1 s'}{w_2 s'}$, $ r_1\eq^* \pair{a_1
        s'}{w_2 s'}$ and $ r_2\eq^* w_2 s'$. So we take $ u_1=( \pair{a_1}{a_2}) t'$ and
      $ u_2= w_2 t'$, the symmetric cases are analogous.
    \item[(c)] $ r_1\eq^* a_1 s'$ and $ r_2\eq^* a_2 s'$, in which case we take
      $ u_1= a_1 t'$ and $ u_2= a_2 t'$ the symmetric case is analogous.
    \end{enumerate}
  \item $ s=\Lambda X. a$, $ t=\Lambda X. t'$, and $ a\re t'$, with $ a\eq^*
    \pair{a_1}{a_2}$ and $ s\eq^*\Lambda X. \pair{a_1}{a_2} X. a_2$. Therefore, by
    the induction hypothesis, there exists $ u'_1$, $ u'_2$ such that either ($
    a_1\toreq u'_1$ and $ a_2\toreq u'_2$), or ($ a_1\eq^* u'_1$ and $ a_2\toreq
    u'_2$), or ($ a_1\toreq u'_1$ and $ a_2\eq^* u'_2$). Therefore, we take $
    u_1=\Lambda X. u_1'$ and $ u_2=\Lambda X. u_2'$.
  \item $ s= \tapp{a}{A}$, $ t= t'[A]$, and $ a\re t'$, with $ a\eq^*
    \pair{a_1}{a_2}$ and $ s\eq^* \pair{a_1[A]}{a_2[A]}$. Therefore, by the
    induction hypothesis, there exists $ u'_1$, $ u'_2$ such that either ($
    a_1\toreq u'_1$ and $ a_2\toreq u'_2$), or ($ a_1\eq^* u'_1$ and $ a_2\toreq
    u'_2$), or ($ a_1\toreq u'_1$ and $ a_2\eq^* u'_2$). Therefore, we take $ u_1=
    u_1'[A]$ and $ u_2= u_2'[A]$. \qedhere
  \end{itemize}
\end{proof}

\section{Detailed proofs of Section~\ref{sec:Adequacy}}\label{app:Adequacy}
\xrecap{Lemma}{Adequacy of product}{lem:adequacyOfProd}{
  For all $r, s, A, B$ such that $r\in\interp A$ and $s\in\interp B$, we have $
  \pair{r}{s}\in\interp{A\wedge B}$.
}
\begin{proof}
  We need to prove that $\ka{}{X}{A\wedge B}{{\pair{r}{s}}}\in\SN$.
  We proceed by induction on the number of projections in $\kk{}{X}{{A\wedge B}}$.
  Since the hole of
$\kk{}{X}{{A\wedge B}}$ has type $A\wedge B$, and $\ka{}{X}{A\wedge B}{{t}}$ has type
$X$ for any $t$ of type $A\wedge B$, we can assume, without lost of generality, that the context
$\kk{}{X}{{A\wedge B}}$ has the form $\ka{'}{X}{C}{\pi_C (\hole{A\wedge B}\alpha_1\dots
\alpha_n)}$, where each $\alpha_i$ is either a term or a type argument.
  We prove that all $\ka{'}{X}{C}{\pi_C (\pair{r\alpha_1\dots \alpha_n}{s\alpha_1\dots \alpha_n})}\in\SN$
  by showing, more generally, that if $r'$ and $s'$ are two reducts of
  $r\alpha_1\dots \alpha_n$ and $s\alpha_1\dots \alpha_n$, then
  $\ka{'}{X}{C}{\pi_C(\pair{r'}{s'})}\in\SN$. For this, we show that all its one step reducts are
  in $\SN$, by induction on $|\kk{'}{X}{C}| + |r'|+|s'|$.
  \begin{itemize}
  \item If the reduction takes place in one of the terms in $\T(\kk{'}{X}{C})$, in
    $r'$, or in $s'$, we apply the induction hypothesis.
  \item Otherwise, the reduction is a \rulelabel{($\pi$)} reduction of $\pi_C
    (\pair{r'}{s'})$, that is, $\pair{r'}{s'}\eq^*\pair{v}{w}$, the reduct is $v$, and we need to prove
    $\ka{'}{X}{C}{v}\in\SN$.
    By Lemma~\ref{lem:eqProd}, we have either:
    \begin{itemize}
    \item $v\eq^* \pair{r_1}{s_1}$, with $r'\eq^* \pair{r_1}{r_2}$ and
      $s'\eq^* \pair{s_1}{s_2}$. In such a case,
      by Lemma~\ref{lem:SNimpliespiSN}, $v$ is the product of two reducible terms, so since there is one projection less than in $\kk{}{X}{{A\wedge B}}$, the first induction hypothesis applies.
    \item $v\eq^*\pair{r'}{s_1}$, with $s'\eq^*\pair{s_1}{s_2}$. In such a case,
      by Lemma~\ref{lem:SNimpliespiSN}, $v$ is the product of two reducible terms, so since there is one projection less than in $\kk{}{X}{{A\wedge B}}$, the first induction hypothesis applies.
    \item $v\eq^*\pair{r_1}{s'}$, with $r'\eq^*\pair{r_1}{r_2}$. In such a case,
      by Lemma~\ref{lem:SNimpliespiSN}, $v$ is the product of two reducible terms, so since there is one projection less than in $\kk{}{X}{{A\wedge B}}$, the first induction hypothesis applies.
    \item $v\eq^*r'$, in which case, $C\equiv A$, and since $r\in\interp A$, we have $\ka{'}{X}{A}{r'}\eq^*\ka{'}{X}{A}{v}\in\SN$.
    \item $v\eq^*r_1$ with $r'\eq^* \pair{r_1}{r_2}$, in which case, since $r\in\interp A$, we have $\ka{'}{X}{C}{\pi_C(r')}\in\SN$ and $\ka{'}{X}{C}{\pi_C(r')}\toreq \ka{'}{X}{C}{v}$ hence $\ka{'}{X}{C}{v}\in\SN$.
    \item $v\eq^*s'$, in which case, $C\equiv B$, and since $s\in\interp B$, we have $\ka{'}{X}{B}{s'}\eq^*\ka{'}{X}{B}{v}\in\SN$.
    \item $v\eq^*s_1$ with $s'\eq^* \pair{s_1}{s_2}$, in which case, since $s\in\interp B$, we have $\ka{'}{X}{C}{\pi_C(s')}\in\SN$ and $\ka{'}{X}{C}{\pi_C(s')}\toreq \ka{'}{X}{C}{v}$ hence $\ka{'}{X}{C}{v}\in\SN$.
      \qedhere
    \end{itemize}
  \end{itemize}
\end{proof}

\xrecap{Lemma}{Adequacy of type abstraction}{lem:tlambda}
{
For all $r$, $X$, $A$, $B$ such that $\substa{X}{B}{r} \in \interp{\substa{X}{B}{A}}$,
we have $\Lambda X . r \in \interp{\tfor{X}{A}}$.
}
\begin{proof}
  We proceed by induction on $M(r)$.
  \begin{itemize}
  \item 
    If $r\eq^* \pair{r_1}{r_2}$, then by
    Lemma~\ref{psi-gen}, then $A\equiv A_1\wedge A_2$ with $r_1$ of type $A_1$ and $r_2$ of type $A_2$,
    and so by Lemma~\ref{substitution},
    $\substa{X}{B}{r_1}$ has type $\substa{X}{B}{A_1}$ and $\substa{X}{B}{r_2}$ has type $\substa{X}{B}{A_2}$. Since $\substa{X}{B}{r}\in\interp{\substa{X}{B}{A}}$, we have $\pair{\substa{X}{B}{r_1}}{\substa{X}{B}{r_2}}\in\interp{\substa{X}{B}{A}}$.
    By Lemma~\ref{lem:SNimpliespiSN}, $\substa{X}{B}{r_1}\in\interp{\substa{X}{B}{A_1}}$ and
    $\substa{X}{B}{r_2}\in\interp{\substa{X}{B}{A_2}}$. By the induction hypothesis, $\Lambda X.r_1\in\interp{\tfor{X}{A_1}}$ and $\Lambda X.r_2\in\interp{\tfor{X}{A_2}}$, then by Lemma~\ref{lem:adequacyOfProd}, 
    $\Lambda X.r\eq^*\pair{\Lambda X.r_1}{\Lambda X.r_2}\in\interp{(\tfor{X}{    A_1})\wedge(\tfor{X}{A_2})}$, and by Lemma~\ref{lem:eqInterp},
    $\interp{(\tfor{X}{A_1})\wedge(\tfor{X}{A_2})}=\interp{\tfor{X}{A}}$.

  \item If $r\nneq^*  \pair{r_1}{r_2}$, we need to prove that for any elimination context $\kk{}{X}{{\tfor{X}{A}}}$, we have
    $\ka{}{X}{{\tfor{X}{A}}}{\Lambda X. r}\in\SN$.
    
    Since $r$ and all the terms in $\T(\kk{}{X}{{\tfor{X}{A}}})$ are in \SN, we proceed by induction on $|\kk{}{X}{{\tfor{X}{A}}}| + |r|$ to show that all
    the one step reducts of $\ka{}{X}{{\tfor{X}{A}}}{\Lambda X.r}$ are in $\SN$.
    Since $r$ is not a product, its only one step reducts are the following.
    \begin{itemize}
    \item If the reduction takes place in one of the terms in $\T(\kk{}{X}{{\tfor{X}{A}}})$ or $r$, we apply the
      induction hypothesis.
    \item If $\ka{}{X}{{\tfor{X}{A}}}{\Lambda X.r} = \ka{'}{X}{{A}}{\tapp{(\Lambda X.r)}{B}}$ and it reduces to $\ka{'}{X}{A}{\substa{X}{B}{r}}$, as $\substa{X}{B}{r}\in\interp A$, we have
      $\ka{'}{X}{A}{ \substa{X}{B}{r}}\in\SN$.
    \item If $r \eq^* \lambda x^C . r'$, then by
      Lemma~\ref{psi-gen}, we have $A \equiv \timpl{C}{A'}$ with $r'$ of type $A'$ and $X \not \in \fvt{C}$,
      and so by Lemma~\ref{substitution},
      $\substa{X}{B}{r'}$ has type $\substa{X}{B}{A'}$. By Lemma~\ref{lem:var}, we have $x \in \interp{\substa{X}{B}{C}}$, and since $\substa{X}{B}{\lambda x^C . r'}\in\interp{\subst{X}{B}{\timpl{C}{A'}}}$, by Lemma~\ref{lem:appOfInterp}, we have $\substa{X}{B}{r'}\in\interp{\substa{X}{B}{A'}}$. By the induction hypothesis, we have $\Lambda X.r'\in\interp{\tfor{X}{A'}}$, then by Lemma~\ref{lem:lamOfInterp}, 
      $\Lambda X.r\eq^*\Lambda X.\lambda x^C.r'\in\interp{\timpl{C}{\tfor{X}{A'}}}$, and by Lemma~\ref{lem:eqInterp}, since $X \not \in \fvt{C}$, we have
      $\interp{\timpl{C}{\tfor{X}{A'}}}=\interp{\tfor{X}{A}}$.
      \qedhere
    \end{itemize}
  \end{itemize}
\end{proof}

\end{document}